\documentclass[10pt,twocolumn,twoside]{IEEEtran}
\usepackage{amsfonts}
\usepackage{cite,graphicx,amsmath,amsthm}
\usepackage{subfigure}
\usepackage{fancyhdr}
\usepackage{dsfont}
\usepackage{array,color}
\usepackage{bm}
\usepackage{float}
\usepackage{graphicx}

\newtheorem{lemma}{Lemma}

\newtheorem{proposition}{Proposition}

\newtheorem{corollary}{Corollary}

\newtheorem{property}{Property}

\newtheorem{remark}{Remark}

\newtheorem{claim}{Claim}

\begin{document}
\title{{Multi-Pair Two-Way Relay Network with Harvest-Then-Transmit Users: \\Resolving Pairwise Uplink-Downlink Coupling}}

\author{Shuai Wang, Minghua Xia,
        and Yik-Chung Wu
\thanks{
Copyright (c) 2016 IEEE. Personal use of this material is permitted. However, permission to use this material for any other purposes must be obtained from the IEEE by sending a request to pubs-permissions@ieee.org.

S. Wang and Y.-C. Wu are with the Department of Electrical and Electronic Engineering, The University of Hong Kong, Hong Kong (e-mail: \{swang, ycwu\}@eee.hku.hk).

M. Xia is with the School of Electronics and Information Technology, Sun Yat-sen University, Guangzhou, 510006, China (e-mail: xiamingh@mail.sysu.edu.cn).

This work was presented in part at IEEE International Conference on Accoustic, Speech and Signal Processing (ICASSP) [32], Shanghai, China, March 2016.
This work was supported by the National Natural Science Foundation of China (NSFC) under Grant No. 61671488, Special Fund for Science and Technology Development in Guangdong Province under Grant No. 2016A050503025, and Guangzhou Science and Technology Project under Grant No. 201604010073.
}}

\maketitle

\begin{abstract}
While two-way relaying is a promising way to enhance the spectral efficiency of wireless networks, the imbalance of relay-user distances may lead to excessive wireless power at the nearby-users.
To exploit the excessive power, the recently proposed harvest-then-transmit technique can be applied.
However, it is well-known that harvest-then-transmit introduces uplink-downlink coupling for a user.
Together with the co-dependent relationship between paired users and interference among multiple user pairs, wirelessly powered two-way relay network suffers from the unique pairwise uplink-downlink coupling, and the joint uplink-downlink network design is nontrivial.
To this end,
for the one pair users case,
we show that a global optimal solution can be obtained.
For the general case of multi-pair users,
based on the rank-constrained difference of convex program,
a convergence guaranteed iterative algorithm with an efficient initialization is proposed.
Furthermore, a lower bound to the performance of the optimal solution is derived by introducing virtual receivers at relay.
Numerical results on total transmit power show that the proposed algorithm achieves a transmit power value close to the lower bound.
\end{abstract}

\begin{IEEEkeywords}
Beamforming design, convex optimization, difference of convex program, multi-pair two-way relay network, rank relaxation, wireless power transfer.
\end{IEEEkeywords}

\IEEEpeerreviewmaketitle
\section{Introduction}

\subsection{Background and Motivation}
Wireless power transfer (WPT) is a groundbreaking technique which can prolong the lifetime of battery and feed power to a device when wired charging is inconvenient \cite{1,J2}.
However,
a critical challenge for WPT is the high propagation path-loss during energy transmission \cite{2}.
Fortunately, the beamforming gain in multiple-input-multiple-output (MIMO) systems \cite{J4} offers a viable option for mitigating such problem \cite{3}.
Therefore, WPT combined with MIMO systems have been a focus lately.
Examples include multi-user multiple-input-single-output channel \cite{5}, MIMO broadcast channel \cite{J1}, and distributed antenna network \cite{J3}.
In all these systems, the interference takes on a dual role: it is detrimental to information decoders, but friendly to energy harvesters.

In the context of WPT, the harvest-then-transmit protocol is recently proposed in \cite{8}, which enables energy constrained terminals to transmit data, and ignites the researches on wirelessly powered communication networks (WPCNs) \cite{8,11,12}.
In WPCNs, the access point supplies power through the downlink WPT, and the harvested energy at users supports subsequent uplink transmission.
This mechanism reveals its uplink-downlink coupled nature, which is apparent in existing works on WPCNs.
In particular in \cite{8}, the throughput versus time allocation of uplink-downlink transmission is analyzed in WPCN.
Meanwhile, the joint uplink-downlink beamforming design in multi-antenna WPCN is proposed in \cite{11}.

However, the current research on WPCNs is still at its infancy, focusing on basic point-to-point systems.
This paper will take a step further to discuss the wirelessly powered multi-pair two-way relay (TWR) network, which
is a natural generalization of many basic relay systems \cite{13,14}.
Wirelessly powered TWR has extensive applications in sensor networks, medical electronics, smart homes and wearable computations etc., where information exchange between energy-constrained devices is often required.
Nonetheless, wirelessly powered TWR is fundamentally different from traditional TWR \cite{13,14,17,15} or other WPT systems \cite{3,5,J1,J3,8,11} due to the \emph{pairwise uplink-downlink coupling}.
More specifically, in wirelessly powered TWR, a single transmit beamformer at the relay simultaneously controls a pair of users' downlink information and power transfer.
Further adding to the fact that the uplink signal-to-interference-plus-noise ratios (SINRs) of a pair of users are determined by a single receive beamformer, and possibly the power harvested in the previous downlink phase, wirelessly powered TWR suffers from the strong coupling of four transmission links for each pair of users.

Under such pairwise uplink-downlink coupling, an adjustment appearing beneficial to a user may \emph{backfire}.
For example, consider aligning the downlink transmit beamformer towards a user $U_1$ and away from its paired user $U_2$.
While this seems to benefit the power transfer to $U_1$, the subsequent uplink transmit power of $U_2$ may decrease due to its reduced harvested energy.
Then, to maintain the uplink data-rate of $U_2$, the receive beamformer needs to be aligned towards $U_2$ and away from $U_1$.
This will in turn deteriorate the uplink data-rate of $U_1$ and the transmit power of $U_1$ needs to be increased, which offsets the benefit brought by the increased harvested power at $U_1$.
Consequently, it is crucial to strike a balance where wirelessly powered TWR is the most efficient.

\subsection{Technical Challenges and Contributions}
While finding the balance can be cast as a beamforming design problem subject to data-rate quality-of-service (QoS) constraints, the pairwise uplink-downlink coupling makes this problem nonconvex over each beamformer.
Therefore, convex programming \cite{3,5,J1} or block coordinate descent method \cite{11,12} will no longer apply.
On the other hand, to resolve the nonconvexity, a common way is to reformulate the problem into a difference of convex (DC) program.
However, obtaining an exact DC reformulation is a usual challenge.
By analyzing the stationary points, this paper proposes a simplification procedure without approximation.
The simplified problem is equivalently reformulated as a DC program using rank relaxation, provided that there exists a rank guarantee of the solution.
Unfortunately, the considered problem involves a single beamformer serving a group of users, and currently there is no rank guarantee of the solution, even for its special case of TWR without energy harvesting \cite{13}.
To this end, we establish a rank-one guarantee for the one pair users case, and a rank-two guarantee for the general multi-pair users case.
Interestingly, the rank-two guarantee result also provides a general property for other beamforming design problems with single beamformer serving multiple users, such as those in \cite{13} and \cite{16}.

With the transformed DC program, it can be solved by an iterative algorithm with a feasible starting point.
Nonetheless, in the multi-pair users case, finding a feasible point is nontrivial when the number of antennas at relay is not sufficient (e.g., when smaller than the number of users).
While the conventional $l_1$-norm regularization is a popular initialization method, it involves multiple conic programs (CPs).
On the contrary, based on the structure of TWR, this paper proposes a CP-free initialization with much lower complexity.
Finally, a lower bound on the performance of the optimal solution is derived by introducing virtual receivers at relay, and an analytical solution is obtained for the case of massive antenna array at relay.
Numerical results show that the proposed iterative algorithm in multi-pair users case achieves close performance to the lower bound.

\subsection{Organization and Notations}
The rest of the paper is organized as follows. In Section II, the system model is described in detail for the uplink and downlink transmission.
In Section III, the beamforming design problem is formulated, and transformed to an equivalent problem for subsequent processing.
In Section IV, the global optimal solution for the one-pair users case is derived in part A, and the convergence guaranteed iterative algorithm is presented for the multi-pair users case in part B.
A lower bound to the performance of the optimal solution and the analytical solution for the case of massive antenna array at relay are derived in Section V.
Simulation results are presented in Section VI.
Finally, conclusions are drawn in Section VII.

\emph{Notation}.
Italic letters, simple bold letters, and capital bold letters represent scalars, vectors, and matrixes, respectively.
The operators $\textrm{Tr}(\cdot),(\cdot)^{T},(\cdot)^{H},\mathrm{Rank}(\cdot),(\cdot)^{-1}$ take the trace, transpose, Hermitian, rank, and inverse of a matrix, respectively,
while $\mathrm{vec}(\cdot)$ is the matrix vectorization operator.
Symbol $\mathbf{I}_{N}$ represents an $N\times N$ identity matrix.
Symbol $\mathbf{1}_{N}$ represents an $N\times 1$ vector with all the elements equal to $1$.
Finally, $\mathbb{E}(\cdot)$ represents the expectation of a random variable.

\section{System Model}

\begin{figure}[!t]
\centering
\includegraphics[width=75mm]{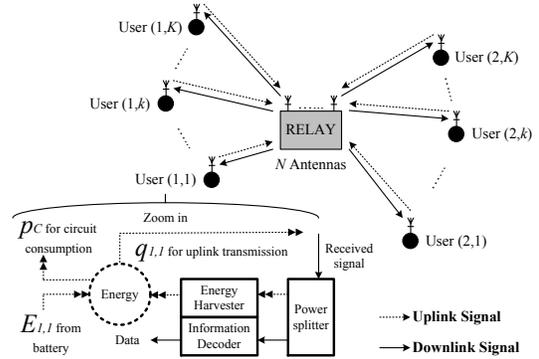}
\caption{System model of wirelessly powered multi-pair TWR network.}
\label{fig_sim}
\end{figure}

\setcounter{secnumdepth}{4}In this paper, we consider a multi-pair TWR network consisting of a relay station with $N$ antennas, and $2K$ single-antenna users.
As shown in Fig. 1, the $2K$ users are paired into $K$ groups, and users in each group intend to exchange information with each other through the relay.
Without loss of generality, users in the $k^{\mathrm{th}}$ group are denoted as user $(1,k)$ and $(2,k)$, where $k\in\mathcal{K}=\{1,2,...,K\}$.
Compute-and-forward network coding with lattice codes is applied in the system as it outperforms amplify-and-forward scheme \cite{17}, and has a lower complexity than decode-and-forward scheme \cite{15}.
The system operates in the half-duplex mode, and the transmission consists of two phases: uplink phase and downlink phase, each with time duration of $M$ symbols.
Below, we give the details of each transmission phase.

\subsection{Uplink Phase}

In uplink phase, all the users transmit data symbols to the relay simultaneously, with the $(i,k)^{\mathrm{th}}$ user symbol being $\mathbf{x}^T_{i,k}\in \mathbb{C}^{1\times M}$ with power $\frac{1}{M}\mathbb{E}(||\mathbf{x}_{i,k}||^2)=q_{i,k}$.
The $\mathbf{x}_{i,k}$ is generated from lattice codes \cite{17} and the details are given in Appendix A.
The received signal $\mathbf{Y}\in \mathbb{C}^{N\times M}$ at the relay is given by
\begin{eqnarray}\label{1}
\mathbf{Y}=\mathop{\sum}_{k=1}^{K}(\mathbf{h}_{1,k}\mathbf{x}^T_{1,k}+\mathbf{h}_{2,k}\mathbf{x}^T_{2,k})+\mathbf{N},
\end{eqnarray}
where $\mathbf{h}_{i,k}\in \mathbb{C}^{N\times 1}$ denotes the uplink channel vector, and $\mathbf{N}\in \mathbb{C}^{N\times M}$ is the Gaussian noise at relay with $\mathbb{E}[\mathrm{vec}(\mathbf{N})\mathrm{vec}(\mathbf{N})^H]=\sigma_r^2\mathbf{I}_{MN}$.
In order to separate the signal component of the $k^{\mathrm{th}}$ user pair, a receive beamforming vector $\mathbf{w}^H_k\in \mathbb{C}^{1\times N}$ with $||\mathbf{w}_k||=1$ is applied to $\mathbf{Y}$.
Since in network coding, $\mathbf{w}^H_k(\mathbf{h}_{1,k}\mathbf{x}^T_{1,k}+\mathbf{h}_{2,k}\mathbf{x}^T_{2,k})$ is considered to be the useful signal for the $k^{\mathrm{th}}$ user pair, and the remaining part is the interference, we can express the uplink SINR of the $(i,k)^{\mathrm{th}}$ user as
\begin{align}
&
\Gamma^{\mathrm{UL}}_{i,k}=\frac{q_{i,k}|\mathbf{w}^H_k\mathbf{h}_{i,k}|^2}{\sum_{j}\sum_{l\neq k}q_{j,l}|\mathbf{w}^H_k\mathbf{h}_{j,l}|^2+\sigma_r^2}.
\end{align}
Furthermore, applying the results from \cite[Theorem 3]{17}, the uplink achievable rate $R^{\mathrm{UL}}_{i,k}$ from the $(i,k)^{\mathrm{th}}$ user to the relay can be computed to be
\begin{align}\label{3}
&R^{\mathrm{UL}}_{i,k}=
\frac{1}{2}{\Big[\mathrm{log}\Big(\frac{q_{i,k}|\mathbf{w}^H_{k}\mathbf{h}_{i,k}|^2}{q_{1,k}|\mathbf{w}^H_{k}\mathbf{h}_{1,k}|^2+q_{2,k}|\mathbf{w}^H_{k}\mathbf{h}_{2,k}|^2}
+\Gamma^{\mathrm{UL}}_{i,k}\Big)\Big]}^+,
\end{align}
where $[x]^+=\mathrm{max}(x,0)$, and the factor $\frac{1}{2}$ is due to the fact that two transmission phases are involved for single symbol transmission.
\subsection{Downlink Phase with WPT}

In the downlink phase, by using $\mathbf{w}^H_k\mathbf{Y}$, the relay generates lattice symbols $\mathbf{s}^T_{k}\in\mathbb{C}^{1\times M}$ with power
$\frac{1}{M}\mathbb{E}(||\mathbf{s}_{k}||^2)=p_{k}$ (details documented in Appendix A).
Then the relay transmits $\mathbf{s}_{k}$ to the $k^{\mathrm{th}}$ user pair through the corresponding transmit beamforming vector $\mathbf{v}_k\in\mathbb{C}^{N\times 1}$ with $||\mathbf{v}_k||=1$.
Accordingly, the received signal $\mathbf{r}^T_{i,k}\in \mathbb{C}^{1\times M}$ at the user $(i,k)$ is
\begin{align}\label{5}
&\mathbf{r}^T_{i,k}
=\mathbf{g}^H_{i,k}\Big(\mathop{\sum}_{l=1}^K\mathbf{v}_{l}\mathbf{s}^T_{l}\Big)+\mathbf{n}^T_{i,k},
\end{align}
where $\mathbf{g}^H_{i,k}\in \mathbb{C}^{1\times N}$ is the downlink channel vector from the relay to the $(i,k)^{\mathrm{th}}$ user,
and $\mathbf{n}^T_{i,k}\in \mathbb{C}^{1\times M}$ is the Gaussian noise at the $(i,k)^{\mathrm{th}}$ user with $\mathbb{E}[\mathbf{n}_{i,k}\mathbf{n}_{i,k}^H]=\sigma_u^2\mathbf{I}_{M}$.
The received signal \eqref{5} at the $(i,k)^{\mathrm{th}}$ user in the downlink is further split into two branches, one for the information decoder and the other for the energy harvester.

At the information decoder side, the signal is given by
\begin{align}\label{BC}
&~\mathbf{r'}^T_{i,k}
=\sqrt{\beta_{i,k}}\mathbf{g}^H_{i,k}\mathbf{v}_k\mathbf{s}^T_{k}
+\sqrt{\beta_{i,k}}\mathbf{g}^H_{i,k}\Big(\mathop{\sum}_{l\neq k}\mathbf{v}_{l}\mathbf{s}^T_{l}\Big)
\nonumber\\
&~~~~~~~~~~
+\sqrt{\beta_{i,k}}\mathbf{n}^T_{i,k}+\mathbf{z}^T_{i,k},
\end{align}
where $\beta_{i,k}\in(0,1]$ is the splitting factor, and $\mathbf{z}^T_{i,k}\in\mathbb{C}^{1\times M}$ is Gaussian noise introduced by the power splitter, with $\mathbb{E}[\mathbf{z}_{i,k}\mathbf{z}_{i,k}^H]=\sigma_z^2\mathbf{I}_{M}$.
In \eqref{BC}, the first term is the desired network coded signal, the second term is the inter-pair interference, and the last two terms are Gaussian noises.
Based on this interpretation of \eqref{BC}, the downlink SINR for the $(i,k)^{\mathrm{th}}$ user is
\begin{align}\label{9}
&
\Gamma^{\mathrm{DL}}_{i,k}=\frac{\beta_{i,k}p_{k}|\mathbf{g}^H_{i,k}\mathbf{v}_k|^2}
{\beta_{i,k}\sum_{l\neq k}p_{l}|\mathbf{g}^H_{i,k}\mathbf{v}_{l}|^2+\beta_{i,k}\sigma_u^2+\sigma_z^2}.
\end{align}
Then the downlink achievable rate $R^{\mathrm{DL}}_{i,k}$ from the relay to the $(i,k)^{\mathrm{th}}$ user is
\begin{align}
&R^{\mathrm{DL}}_{i,k}=
\frac{1}{2}\mathrm{log}\Big(1+\Gamma^{\mathrm{DL}}_{i,k}\Big).
\end{align}

On the other hand, the average harvested power from the wireless signals at user $(i,k)$ can be expressed as
$\eta(1-\beta_{i,k})\mathbb{E}[||\mathbf{r}_{i,k}||^2]/M$, where $0<\eta<1$ is the power conversion efficiency.
Based on \eqref{5}, it can be further expressed as
$
\eta(1-\beta_{i,k})\Big(\sum_{l=1}^Kp_{l}|\mathbf{g}^H_{i,k}\mathbf{v}_{l}|^2+\sigma_u^2\Big)
$.

\section{Problem Formulation and Transformation}
In the considered network, the design variables that can be controlled are the relay transmit-receive beamformers $\{\mathbf{v}_k,\mathbf{w}_k\}$, relay power allocations $\{p_k\}$, users' transmit powers $\{q_{i,k}\}$, and users' power splitting ratios $\{\beta_{i,k}\}$.
Since a basic QoS requirement in a communication system is the guaranteed transmission rate \cite{5}, our aim is to provide reliable communication for all the users at their required data-rates.
In particular, assuming that the data-rate requirement from the $(i,k)^{\mathrm{th}}$ user is $\overline{R}_{i,k}>0$,
then the channel from the $(i,k)^{\mathrm{th}}$ user to the relay, and that from relay to its paired $(3-i,k)^{\mathrm{th}}$ user\footnote{
The index $3-i$ is used to represent the paired user of the $i^{\mathrm{th}}$ user due to $i\in\{1,2\}$.} should both have achievable rates larger than $\overline{R}_{i,k}$, i.e., $R^{\mathrm{UL}}_{i,k}\geq\overline{R}_{i,k}$ and $R^{\mathrm{DL}}_{3-i,k}\geq\overline{R}_{i,k}$.

On the other hand, since the uplink achievable rate $R^{\mathrm{UL}}_{i,k}$ in \eqref{3} depends on the user transmit power $q_{i,k}$, which is in turn harvested from the downlink wireless signal, we must have the following power constraint\footnote{
To incorporate the minimum power threshold $\epsilon$ for activating the energy harvesting circuit inside a user terminal \cite{R1},
we could add an additional constraint $\eta(1-\beta_{i,k})\Big(\sum_{l=1}^Kp_{l}|\mathbf{g}^H_{i,k}\mathbf{v}_{l}|^2+\sigma_u^2\Big)
\geq\epsilon$. Since this constraint can be reformulated into a convex form, it will not change the subsequent derivation and algorithm.
}
:
\begin{align}
&~
\eta(1-\beta_{i,k})\Big(\sum_{l=1}^Kp_{l}|\mathbf{g}^H_{i,k}\mathbf{v}_{l}|^2+\sigma_u^2\Big)+2E_{i,k}-2p_c
\geq q_{i,k},
\end{align}
where $E_{i,k}$ is the local power per symbol-time at the $(i,k)^{\mathrm{th}}$ user, and $p_c$ is the circuit power consumption per symbol-time (the coefficient $2$ is due to the two phases of transmission).

Having the QoS and power harvesting requirements satisfied, it is crucial to minimize the total transmit power at relay and users\footnote{The user transmit powers are minimized to save energy for the battery.} because saving power translates to cost reduction and environmental benefits.
As a result, by accounting for all the factors mentioned above, an optimization problem can be formulated as:
\begin{align}
&\mathcal{P}1:\mathop{\mathrm{min}}_{\substack{\{\mathbf{v}_k,\mathbf{w}_k,p_k,\\q_{i,k},\beta_{i,k}\}}}
~\sum_{k=1}^K p_{k}+\sum_{i=1}^2\sum_{k=1}^Kq_{i,k} \nonumber\\
&\mathrm{s.t.}~~\frac{q_{i,k}|\mathbf{w}^H_{k}\mathbf{h}_{i,k}|^2}
{\sum_{j=1}^2q_{j,k}|\mathbf{w}^H_{k}\mathbf{h}_{j,k}|^2}+
\frac{q_{i,k}|\mathbf{w}^H_k\mathbf{h}_{i,k}|^2}{\sum_{j=1}^2\sum_{l\neq k}q_{j,l}|\mathbf{w}^H_k\mathbf{h}_{j,l}|^2+\sigma_r^2}
\nonumber\\
&~~~~~~
\geq 2^{2\overline{R}_{i,k}},~~\forall i\in\{1,2\},k\in\mathcal{K}
\nonumber
\\
&~~~~~~1+
\frac{\beta_{i,k}p_{k}|\mathbf{g}^H_{i,k}\mathbf{v}_k|^2}
{\beta_{i,k}\mathop{\sum}_{l\neq k}p_{l}|\mathbf{g}^H_{i,k}\mathbf{v}_{l}|^2+\beta_{i,k}\sigma_u^2+\sigma_z^2}
\geq 2^{2\overline{R}_{3-i,k}},
\nonumber\\
&~~~~~~
\forall i\in\{1,2\},k\in\mathcal{K}
\nonumber \\
&~~~~~~\eta(1-\beta_{i,k})\Big(\sum_{l=1}^Kp_{l}|\mathbf{g}^H_{i,k}\mathbf{v}_{l}|^2+\sigma_u^2\Big)+2E_{i,k}-2p_c\geq q_{i,k},
\nonumber\\
&~~~~~~
\forall i\in\{1,2\},k\in\mathcal{K}
\nonumber\\
&~~~~~~q_{i,k}\geq0,~\beta_{i,k}\in (0,1],~~\forall i\in\{1,2\},k\in\mathcal{K}
\nonumber\\
&~~~~~~||\mathbf{v}_k||=1,~||\mathbf{w}_k||=1,~p_{k}\geq0,~~\forall k\in\mathcal{K},
\nonumber
\end{align}
where the first and second constraints are the uplink and downlink data-rate constraints, respectively.
This optimization problem will be solved at the relay. Then the transmit powers and power splitting ratios will be sent to users.
Unfortunately, problem $\mathcal{P}1$ is extremely difficult since even when there is no uplink phase and energy harvesting, the resultant multigroup multicast problem is NP-hard \cite{13,16}.

The first challenge of solving $\mathcal{P}1$ is the term
$q_{i,k}|\mathbf{w}^H_{k}\mathbf{h}_{i,k}|^2/(\sum_{j=1}^2q_{j,k}|\mathbf{w}^H_{k}\mathbf{h}_{j,k}|^2)$ in the first constraint,
which makes the left hand side nonlinear with respect to $q_{i,k}|\mathbf{w}^H_{k}\mathbf{h}_{i,k}|^2$.
One way to get around this is to consider high SINR regime and set the term to zero as in \cite[Section V-B]{15}.
However, this would yield a conservative solution.
In contrast, we can establish the following property.
\begin{property}
Any stationary point\footnote{Stationary points of a problem are points satisfying the Karush-Kuhn-Tucker condition \cite{30}, and the optimal solution must be a stationary point.} of $\mathcal{P}1$ satisfies
\begin{align}\label{13ar}
\frac{q_{i,k}|\mathbf{w}^H_{k}\mathbf{h}_{i,k}|^2}
{\sum_{j=1}^2q_{j,k}|\mathbf{w}^H_{k}\mathbf{h}_{j,k}|^2}=\frac{2^{2\overline{R}_{i,k}}}{2^{2\overline{R}_{1,k}}+2^{2\overline{R}_{2,k}}},~~\forall i,k.
\end{align}
\end{property}
\begin{proof}
See Appendix B.
\end{proof}
\noindent Putting \eqref{13ar} into $\mathcal{P}1$ and defining
$\alpha_{i,k}:=2^{2\overline{R}_{i,k}}-2^{2\overline{R}_{i,k}}/(2^{2\overline{R}_{1,k}}+2^{2\overline{R}_{2,k}})$, the first constraint of $\mathcal{P}1$ can be replaced by
\begin{align}\label{alpha}
&q_{i,k}|\mathbf{w}^H_k\mathbf{h}_{i,k}|^2\geq\alpha_{i,k}\Big(\sum_{j=1}^2\sum_{l\neq k}q_{j,l}|\mathbf{w}^H_k\mathbf{h}_{j,l}|^2+\sigma_r^2\Big),
\nonumber\\
&\forall i=1,2,~~k=1,...,K,
\end{align}
without changing the solution of $\mathcal{P}1$ \cite{30}.
Now, the constraint \eqref{alpha} is linear with respect to $q_{i,k}|\mathbf{w}^H_{k}\mathbf{h}_{i,k}|^2$.

The second challenge of solving $\mathcal{P}1$ is the nonlinearly coupled variables of $\mathbf{v}_k,\mathbf{w}_k,\beta_{i,k}$, $p_{k}$ and $q_{i,k}$ in the first three constraints.
While other WPT systems \cite{3,5,J1,11} also suffer from the coupled variables,
the problems in those systems are convex or ``hidden convex'' (appear nonconvex but can be reformulated as convex problems) over each variable.
On the contrary, due to the pairwise uplink-downlink coupling, the feasible set of $\mathcal{P}1$ is nonconvex over $\mathbf{v}_k$ and $\mathbf{w}_k$ even with other variables fixed.
Therefore, the convex programming \cite{3,5,J1} or the block coordinate decent algorithm \cite{11} is not applicable to $\mathcal{P}1$.
To this end, we introduce new variables:
\begin{align}
\xi_{i,k}=\frac{1}{q_{i,k}},~\mathbf{V}_k:=p_{k}\mathbf{v}_k\mathbf{v}^H_k\succeq 0,~\mathrm{Rank}(\mathbf{V}_k)=1,
\end{align}
and apply the linear matrix inequalities (LMIs) \cite{20} to reformulate the constraints.
Then, based on the procedure given in Appendix C, the problem $\mathcal{P}1$ is equivalent to the problem $\mathcal{P}2$,
where $\mu_{i,k},I_{j,l,k}$ are slack variables, $\mathbf{\Theta}_{i,k}=\mathbf{g}_{i,k}\mathbf{g}^{H}_{i,k}$, and $\theta_{i,k}=2^{2\overline{R}_{3-i,k}}-1$.
Problem $\mathcal{P}2$ is well-structured with all the constraints being convex except for the left part of \eqref{P4-a}\eqref{P4-mu}, and the rank constraints in \eqref{P4-rank}.
These remaining difficulties will be addressed in the following section.
\begin{subequations}
\begin{align}
&\mathcal{P}2:
\mathop{\mathrm{min}}_{\substack{\{\mathbf{V}_k,\mathbf{w}_k,\xi_{i,k},\beta_{i,k},\\\mu_{i,k},I_{j,l,k}\}}}
~\sum_{k=1}^K \mathrm{Tr}(\mathbf{V}_{k})+\sum_{i=1}^2\sum_{k=1}^K\frac{1}{\xi_{i,k}} \nonumber \\
&\mathrm{s.t.}~~
\alpha_{i,k}\Big(\sum_{j=1}^2\sum_{l\neq k}I_{j,l,k}+\sigma_r^2\Big)-\frac{|\mathbf{w}^H_k\mathbf{h}_{i,k}|^2}{\xi_{i,k}}\leq 0,~~\forall i,k \label{P4-a}
\\
&~~~~~\left[
\begin{array}{cccc}
I_{j,l,k} & \mathbf{w}^H_k\mathbf{h}_{j,l}
\\
\mathbf{h}^H_{j,l}\mathbf{w}_k & \xi_{j,l}
\end{array}
\right]
\succeq 0,~~\forall j\in\{1,2\},\forall k,l\in\mathcal{K},l\neq k\label{P4-a2}
\\
&~~~~~
\left[
\begin{array}{cccc}
\dfrac{\mathrm{Tr}(\mathbf{\Theta}_{i,k}\mathbf{V}_k)}{\theta_{i,k}}-\mathop{\sum}_{l\neq k}\mathrm{Tr}(\mathbf{\Theta}_{i,k}\mathbf{V}_{l})
-\sigma_u^2 & \sigma_z
\\
\sigma_z & \beta_{i,k}
\end{array}
\right]
\nonumber\\
&~~~~~
\succeq 0,~~\forall i,k\label{P4-b}
\\
&~~~~~
\left[
\begin{array}{cccc}
\sum_{l=1}^K\mathrm{Tr}(\mathbf{\Theta}_{i,k}\mathbf{V}_{l})+\sigma_u^2 & \mu_{i,k}
\\
\mu_{i,k} & \eta(1-\beta_{i,k})
\end{array}
\right]
\succeq 0,
\nonumber\\
&~~~~~~\forall i,k\label{P4-c}
\\
&~~~~~~\frac{1}{\xi_{i,k}}-\mu^2_{i,k}+2p_c-2E_{i,k}\leq 0,~~\forall i,k,\label{P4-mu}
\\
&~~~~~~
\mathbf{V}_k\succeq 0,~\mathrm{Rank}(\mathbf{V}_k)=1,~||\mathbf{w}_k||\leq1,~~\forall k. \label{P4-rank}
\end{align}
\end{subequations}

\section{Joint Uplink-Dowlink Network Design}
In this section, the global optimal solution is first derived for the special case $K=1$, and then we step further to the general case $K\geq 2$.
\subsection{One Pair Case: Optimal Solution}
When $K=1$, the network reduces to a three-node TWR.
Since there is no $l\neq k$, the term $I_{j,l,k}$ can be dropped, and \eqref{P4-a} becomes
\begin{align}\label{k=1}
\alpha_{i}\sigma_r^2\xi_{i}-|\mathbf{w}^H\mathbf{h}_{i}|^2\leq 0,~~\forall i\in\{1,2\},
\end{align}
where the subscript $k$ is also dropped to simplify the notation in this subsection.
With a similar proof to \textbf{Property 1},
it can be shown that any stationary point of $\mathcal{P}2$ with $K=1$ must activate \eqref{k=1} and therefore
the problem $\mathcal{P}2$ with $K=1$ is equivalent to
\begin{subequations}
\begin{align}
&\mathcal{P}3:~\mathop{\mathrm{min}}_{\substack{\mathbf{V},\mathbf{w},\{\beta_{i},\mu_i\}}}
~ \mathrm{Tr}(\mathbf{V})+\sum_{i=1}^2\frac{\alpha_{i}\sigma_r^2}{|\mathbf{w}^H\mathbf{h}_{i}|^2} \nonumber \\
&\mathrm{s.t.}~~
\left[
\begin{array}{cccc}
\dfrac{\mathrm{Tr}(\mathbf{\Theta}_{i}\mathbf{V})}{\theta_{i}}
-\sigma_u^2 & \sigma_z
\\
\sigma_z & \beta_{i}
\end{array}
\right]
\succeq 0,~~\forall i
\label{26b}
\\
&~~~~~~
\left[
\begin{array}{cccc}
\mathrm{Tr}(\mathbf{\Theta}_{i}\mathbf{V})+\sigma_u^2 & \mu_i
\\
\mu_i & \eta(1-\beta_{i})
\end{array}
\right]
\succeq 0,~~\forall i
\label{26bb}
\\
&~~~~~~~
\mu_i\geq\sqrt{\Big[\frac{\alpha_{i}\sigma_r^2}{|\mathbf{w}^H\mathbf{h}_{i}|^2}+2p_c-2E_{i}\Big]^+},~~\forall i\label{26c}
\\
&~~~~~~~
\mathbf{V}\succeq 0,~\mathrm{Rank}(\mathbf{V})=1,~||\mathbf{w}||\leq1.
\label{P5-rank}
\end{align}
\end{subequations}
Now to reduce the problem dimension of $\mathcal{P}3$, the following property can be established.
\begin{property}
Denoting $\mathbf{w}^*$ and $\mathbf{V}^*$ as the optimal solution of $\mathbf{w}$ and $\mathbf{V}$ in $\mathcal{P}3$,
we have $\mathbf{w}^*\in\mathrm{span}\{\mathbf{h}_1,\mathbf{h}_2\}$ and the eigenvector of $\mathbf{V}^*$ in $\mathrm{span}\{\mathbf{g}_1,\mathbf{g}_2\}$.
\end{property}
\begin{proof}
See Appendix D.
\end{proof}
Based on \textbf{Property 2} and \cite[Lemma 5-6]{24},
the receive beamformer $\mathbf{w}$ can be parameterized by $0\leq\gamma\leq1$ as
\begin{align}\label{w}
&\mathbf{w}(\gamma)=\sqrt{\gamma}\cdot\frac{\mathbf{h}_{1}}{||\mathbf{h}_{1}||}+\sqrt{1-\gamma}\cdot
\mathrm{e}^{\mathrm{j}\angle(\mathbf{h}^H_2\mathbf{h}_1)}\cdot\frac{\mathbf{h}_b}{||\mathbf{h}_b||}
,
\end{align}
where $\mathrm{j}=\sqrt{-1}$ and $\mathbf{h}_b=
\mathbf{h}_2-
\mathbf{h}^H_1\mathbf{h}_2/||\mathbf{h}_1||^2\cdot\mathbf{h}_1$.
On the other hand, following the procedure in \cite[Section III]{R4},
$\mathbf{V}$ can be expressed as
\begin{align}\label{V}
\mathbf{V}=\mathbf{G}\mathbf{A}\mathbf{G}^H,~\mathrm{Rank}(\mathbf{A})=1,
\end{align}
with $\mathbf{A}\in\mathbb{C}^{2\times2}\succeq 0$.
Putting \eqref{w} and \eqref{V} into $\mathcal{P}3$, we have the following equivalent problem:
\begin{subequations}
\begin{align}
&\mathcal{P}4:~\mathop{\mathrm{min}}_{\substack{\mathbf{A}, \gamma, \{\beta_{i},\mu_i\}}}
~\mathrm{Tr}(\mathbf{G}^H\mathbf{G}\mathbf{A})
+\sum_{i=1}^2\frac{\alpha_{i}\sigma_r^2}{|\mathbf{w}(\gamma)^H\mathbf{h}_{i}|^2}
 \nonumber \\
&\mathrm{s.t.}~~\left[
\begin{array}{cccc}
\dfrac{\mathrm{Tr}(\mathbf{C}_{i}\mathbf{A})}{\theta_{i}}-\sigma_u^2 & \sigma_z
\\
\sigma_z & \beta_{i}
\end{array}
\right]
\succeq 0,~~\forall i\label{P3.2b}
 \\
&~~~~~~\left[
\begin{array}{cccc}
\mathrm{Tr}(\mathbf{C}_{i}\mathbf{A})+\sigma_u^2 & \mu_i
\\
\mu_i & \eta(1-\beta_{i})
\end{array}
\right]
\succeq 0,~~\forall i \label{P3.2c}
\\
&~~~~~~~
\mu_i\geq\sqrt{\Big[\frac{\alpha_{i}\sigma_r^2}{|\mathbf{w}(\gamma)^H\mathbf{h}_{i}|^2}+2p_c-2E_{i}\Big]^+},~~\forall i
\\
&~~~~~~~
\mathbf{A}\succeq 0,~\mathrm{Rank}(\mathbf{A})=1\label{P3.2f},
\end{align}
\end{subequations}
where $\mathbf{C}_{i}=\mathbf{G}^H\mathbf{g}_{i}\mathbf{g}^H_{i}\mathbf{G}$.
In $\mathcal{P}4$, the remaining nonconvex parts are the terms containing $\gamma$ and the rank constraint in \eqref{P3.2f}.
While the nonconvexity of $\gamma$ can be resolved by 1-D search, the rank constraint can be dropped by applying the semidefinite relaxation (SDR) \cite{21}, and the following proposition shows that the relaxation does not affect the optimality.
\begin{proposition}
The optimal rank-one solution $\mathbf{A}^*$ to the SDR problem of $\mathcal{P}4$ always exists.
\end{proposition}
\begin{proof}
See \cite[Proposition 1]{R4}.
\end{proof}
\textbf{Proposition 1} guarantees that an optimal rank-one solution to the SDR problem of $\mathcal{P}4$ exists.
However, there may be an alternative solution with a higher rank.
If we obtain such a solution, the rank reduction procedure in \cite{23} can be applied to obtain the rank-one solution.
Therefore, the problem $\mathcal{P}4$ is equivalent to its SDR problem, which can be optimally solved by an iterative procedure with 1-D search of $\gamma$ over $[0,1]$, and semidefinite programming (SDP) for each fixed $\gamma$.
Denoting the optimal solution to $\mathcal{P}4$ as $\mathbf{A}^*,\gamma^*,\{\beta^*_i,\mu^*_i\}$, then the optimal solution to $\mathcal{P}3$ is
$\mathbf{V}^*=\mathbf{G}\mathbf{A}^*\mathbf{G}^H$ and $\mathbf{w}^*=\mathbf{w}(\gamma^*)$.
Therefore, the optimal $\mathbf{v}^*,\mathbf{w}^*,p^*,\{q^*_i,\beta^*_i\}$ of $\mathcal{P}1$ with $K=1$ can be recovered accordingly.

In terms of computation complexity of the proposed algorithm,
calculating $\mathbf{C}_{i}=\mathbf{G}^H\mathbf{g}_{i}\mathbf{g}^H_{i}\mathbf{G}$ requires complexity $O(2N^2)$ \cite{20}.
On the other hand, solving the SDR problem of $\mathcal{P}4$ is dominated by the $2$-dimensional SDP cone in \eqref{P3.2f}, which requires complexity $O(2^{3.5}t)$ \cite{20}, with $t$ being the number of iterations for one-dimensional search.
As a result, the total computational complexity would be $O(2N^2+2^{3.5}t)$.
Notice that the problem $\mathcal{P}1$ with $K=1$ has also been discussed in \cite{12}, where a block coordinate descent method with each iteration solving an SDP problem was proposed.
However, such a method only converges to a suboptimal solution for the nonconvex problem $\mathcal{P}1$, in contrast to the global optimal solution derived in this paper.
Moreover, since the problem dimension in \cite{12} is $N$, the complexity of the method in \cite{12} is $O(2sN^{3.5})$, with $s$ being the number of iterations,
and is larger than that of the proposed method in this paper.

\subsection{General Case of $K\geq 2$}
For $K\geq 2$, a critical issue of $\mathcal{P}2$ is the rank constraints in \eqref{P4-rank}.
In the following, by analyzing the Lagrangian and the Karush-Kuhn-Tucker (KKT) condition,
a proposition is established for the rank-relaxed problem of $\mathcal{P}2$.
\begin{proposition}
Any stationary point $\mathbf{V}^\diamond_k$ of the rank-relaxed problem of $\mathcal{P}2$ satisfies $\mathrm{Rank}(\mathbf{V}^\diamond_k)\leq 2$.
\end{proposition}
\begin{proof}
To prove the proposition, define a set of users $\mathcal{G}=\{i,k: \beta_{i,k}=1, i\in\{1,2\},k\in\mathcal{K}\}$ who do not need energy harvesting.
First we construct the augmented Lagrangian with respect to $\mathbf{V}_k$ and $\mathbf{\beta}_{i,k}$ for
the rank relaxed problem of $\mathcal{P}2$ as \cite{22}
\begin{align}
&\mathcal{L}=
\sum_{k=1}^K\mathrm{Tr}(\mathbf{V}_{k})
-\sum_{k=1}^K\mathrm{Tr}(\mathbf{\Xi}_{k}\mathbf{V}_{k})
\nonumber\\
&
+\sum_{i=1}^2\sum_{k=1}^K\lambda_{i,k}\Big(
\frac{\sigma_z^2}{\beta_{i,k}}
-\frac{\mathrm{Tr}(\mathbf{\Theta}_{i,k}\mathbf{V}_{k})}{\theta_{i,k}}
+
\sum_{l\neq k}\mathrm{Tr}(\mathbf{\Theta}_{i,k}\mathbf{V}_l)+\sigma_u^2
\Big)
\nonumber\\
&+\sum_{(i,k)\notin\mathcal{G}}\rho_{i,k}\Big(\frac{\mu^2_{i,k}}{\eta(1-\beta_{i,k})}-
\sum_{l=1}^K\mathrm{Tr}(\mathbf{\Theta}_{i,k}\mathbf{V}_l)-\sigma_u^2\Big)
\nonumber\\
&
+\sum_{i=1}^2\sum_{k=1}^K\nu_{i,k}(\beta_{i,k}-1)-\sum_{i=1}^2\sum_{k=1}^K\tau_{i,k}\beta_{i,k},
\end{align}
where $\mathbf{\Xi}_k\succeq 0,\lambda_{i,k}\geq 0,\rho_{i,k}\geq 0,\nu_{i,k}\geq 0,\tau_{i,k}\geq 0$ are Lagrangian multipliers.
According to \cite{30}, the primal variables and the Lagrangian multipliers of any stationary point should together satisfy the KKT condition
$
\partial\mathcal{L}/\partial\mathbf{V}_k=0$ for any $k$ and $\partial\mathcal{L}/\partial \beta_{i,k}=0$ for any $(i,k)\notin\mathcal{G}$,
which can be explicitly expressed as
\begin{align}
&\mathbf{\Upsilon}
-\lambda^\diamond_{1,k}\Big(1+\frac{1}{{\theta_{1,k}}}\Big)\mathbf{\Theta}_{1,k}
-\lambda^\diamond_{2,k}\Big(1+\frac{1}{{\theta_{2,k}}}\Big)\mathbf{\Theta}_{2,k}
=\mathbf{\Xi}^\diamond_k \label{Xi},~~\forall k
\\
&-
\frac{\sigma_z^2\lambda^\diamond_{i,k}}{(\beta^{\diamond}_{i,k})^2}
+\frac{(\mu^\diamond_{i,k})^2\rho^\diamond_{i,k}}{\eta(1-\beta^\diamond_{i,k})^2}
+\nu^\diamond_{i,k}-\tau^\diamond_{i,k}=0,~~\forall (i,k)\notin\mathcal{G},
\label{B2}
\end{align}
where $\mathbf{\Upsilon}=\mathbf{I}
+\sum_{j=1}^2\sum_{l=1}^K\lambda^\diamond_{j,l}\mathbf{\Theta}_{j,l}
-\sum_{(j,l)\notin\mathcal{G}}\rho^\diamond_{j,l}\mathbf{\Theta}_{j,l}$.
On the other hand, with complementary slackness, we immediately have $\nu^\diamond_{i,k}=\tau^\diamond_{i,k}=0$ for all $(i,k)\notin\mathcal{G}$.
From \eqref{B2} and $\nu^\diamond_{i,k}=\tau^\diamond_{i,k}=0$, the Lagrangian multiplier $\rho^\diamond_{i,k}$ can be further expressed as
\begin{align}\label{rho}
\rho^\diamond_{i,k}=
\frac{\eta(1-\beta^\diamond_{i,k})^2\sigma_z^2}{(\mu^\diamond_{i,k}\beta^{\diamond}_{i,k})^2}\lambda^\diamond_{i,k},~~\forall (i,k)\notin\mathcal{G}.
\end{align}
Notice that $\mu^\diamond_{i,k}\neq0$ for any $(i,k)\notin\mathcal{G}$.
Now, by substituting \eqref{rho} into the definition of $\mathbf{\Upsilon}$, we have
\begin{align}\label{upsilon}
&\mathbf{\Upsilon}=\mathbf{I}
+\mathop{\sum}_{(j,l)\in\mathcal{G}}\lambda^\diamond_{j,l}\mathbf{\Theta}_{j,l}
\nonumber\\
&~~~~~
+\mathop{\sum}_{(j,l)\notin\mathcal{G}}\Big(1-\frac{\eta(1-\beta^\diamond_{j,l})^2\sigma_z^2}{(\mu^\diamond_{j,l}\beta^{\diamond}_{j,l})^2}\Big)
\lambda^\diamond_{j,l}\mathbf{\Theta}_{j,l}.
\end{align}

Below we will show that $\mathbf{\Upsilon}$ is of full rank, i.e., all of its eigenvalues are nonzero.
This can be proved by contradiction.
Assuming that there exists a vector $\mathbf{a}\neq\mathbf{0}$ such that $\mathbf{a}^H\mathbf{\Upsilon}\mathbf{a}=0$, then
by left multiplying $\mathbf{a}^H$ and right multiplying $\mathbf{a}$ to \eqref{Xi},
we have
\begin{align}\label{24}
&\mathbf{a}^H\mathbf{\Upsilon}\mathbf{a}-\lambda^\diamond_{1,k}\big(1+\frac{1}{{\theta_{1,k}}}\big)\mathbf{a}^H\mathbf{\Theta}_{1,k}\mathbf{a}
-\lambda^\diamond_{2,k}\big(1+\frac{1}{{\theta_{2,k}}}\big)\mathbf{a}^H\mathbf{\Theta}_{2,k}\mathbf{a}
\nonumber\\
&
=\mathbf{a}^H\mathbf{\Xi}^\diamond_k\mathbf{a}\geq 0,~~\forall k,
\end{align}
where the inequality is due to $\mathbf{\Xi}^\diamond_k\succeq 0$.
Using $\mathbf{a}^H\mathbf{\Upsilon}\mathbf{a}=0$, \eqref{24} further simplifies to
\begin{align}
&-\sum_{i=1}^2\lambda^\diamond_{i,k}\big(1+\frac{1}{{\theta_{i,k}}}\big)\mathbf{a}^H\mathbf{\Theta}_{i,k}\mathbf{a}
\geq 0.
\label{0=0}
\end{align}
Due to $\mathbf{\Theta}_{i,k}\succeq 0$, we have
$-\lambda^\diamond_{i,k}\big(1+\frac{1}{{\theta_{i,k}}}\big)\mathbf{a}^H\mathbf{\Theta}_{i,k}\mathbf{a}\leq 0$ holds.
Using this result and \eqref{0=0}, we have
$\lambda^\diamond_{i,k}\big(1+\frac{1}{{\theta_{i,k}}}\big)\mathbf{a}^H\mathbf{\Theta}_{i,k}
\mathbf{a}=0$ for any $(i,k)$.
Since $1+\frac{1}{{\theta_{i,k}}}\neq 0$, we have
$\lambda^\diamond_{i,k}\mathbf{a}^H\mathbf{\Theta}_{i,k}
\mathbf{a}=0$ for any $(i,k)$.
As a consequence,
by left multiplying $\mathbf{a}^H$ and right multiplying $\mathbf{a}$ to \eqref{upsilon},
we have
$\mathbf{a}^H\mathbf{\Upsilon}\mathbf{a}=\mathbf{a}^H\mathbf{I}\mathbf{a}=||\mathbf{a}||^2$.
Before \eqref{24}, we assumed that $\mathbf{a}^H\mathbf{\Upsilon}\mathbf{a}=0$, which leads to $\mathbf{a}=\mathbf{0}$.
This contradicts to $\mathbf{a}\neq\mathbf{0}$, and therefore $\mathbf{\Upsilon}$ is of full rank.

Using $\mathrm{Rank}(\mathbf{\Upsilon})=N$ and \eqref{Xi}, and due to $\mathbf{\Theta}_{i,k}$ defined under $\mathcal{P}2$ is of rank one,
we have
$
\mathrm{Rank}(\mathbf{\Xi}^\diamond_k)\geq N-2.
$
On the other hand, the complementary slackness condition for $\mathbf{\Xi}^\diamond_k$ and $\mathbf{V}^\diamond_k$ is
$\mathbf{\Xi}^\diamond_k\mathbf{V}^\diamond_k=\mathbf{0}$ which implies
$
\mathrm{Rank}(\mathbf{\Xi}^\diamond_k\mathbf{V}^\diamond_k)=0.
$
Now, using Sylvester's Inequality, we have
\begin{align}
\mathrm{Rank}(\mathbf{\Xi}^\diamond_k)+\mathrm{Rank}(\mathbf{V}^\diamond_k)-N\leq
\mathrm{Rank}(\mathbf{\Xi}^\diamond_k\mathbf{V}^\diamond_k).
\end{align}
Putting $\mathrm{Rank}(\mathbf{\Xi}^\diamond_k)\geq N-2$ and $\mathrm{Rank}(\mathbf{\Xi}^\diamond_k\mathbf{V}^\diamond_k)=0$ into the above inequality immediately contributes to $\mathrm{Rank}(\mathbf{V}^\diamond_k)\leq 2$ and the proof is completed.
\end{proof}
Notice that by setting $\beta_{i,k}=1$ for all $i,k$, \textbf{Proposition 2} applies to many systems without WPT, such as those in \cite{13} and \cite{16}.
Specifically, \textbf{Proposition 2} can be directly applied to the power minimization problem in multi-pair TWR \cite{13}.
Furthermore, for the multigroup multicast problem in \cite{16}, with the number of users in each group equal to $2$,
\textbf{Proposition 2} holds for the QoS problem, and
with the transformation in \cite[Claim 3]{16}, it also applies to the max-min fairness (MMF) problem.
While \cite{13,16} only claim that the rank-one solution is not guaranteed in their problems, \textbf{Proposition 2} reveals that the solution must be rank two or less.
Finally, when the number of users in each group in the multigroup multicast problem \cite{16} is larger than $2$, with a similar proof as in \textbf{Proposition 2}, it can be shown that $\mathrm{Rank}(\mathbf{V}_k)\leq G_k$ always holds for the QoS and MMF problems, where $G_k$ is the number of users in the $k^{\mathrm{th}}$ group.

\textbf{Proposition 2} indicates that the rank of $\mathbf{V}^\diamond_k$ is either 1 or 2.
When $\mathbf{V}^\diamond_k$ is rank-one, the rank relaxation is tight.
On the other hand, when $\mathbf{V}^\diamond_k$ is rank-two, transmission in blocks of Alamouti codes will also guarantee no performance loss \cite{28}.
Therefore, dropping the rank constraint in \eqref{P4-rank} will not affect the solution of $\mathcal{P}2$.

Problem $\mathcal{P}2$ without the rank constraints in \eqref{P4-rank} is a DC program, since the left hand side of \eqref{P4-a} and \eqref{P4-mu} are difference of convex functions.
DC programs received a lot of attention lately (e.g., in medical imaging, financial engineering, and machine learning \cite{B1}) due to the fact that a local optimal solution can be found in polynomial time using the inner approximation method \cite{25,B1,B2,A1,A6}.
Inner approximation is an iterative method where the DC part is replaced by a convex upper bound expanded around the last round solution.
In particular, define the second term on the left hand side of \eqref{P4-a} as $\Phi_{i,k}(\mathbf{w}_k,\xi_{i,k}):=
-|\mathbf{w}^H_k\mathbf{h}_{i,k}|^2/\xi_{i,k}$.
Assuming that the solution at the $n^{\mathrm{th}}$ iteration is given by $\{\mathbf{w}^{[n]}_k,\xi^{[n]}_{i,k}\}$,
now define another function $\tilde{\Phi}^{[n]}_{i,k}$ as
\begin{align}\label{gapp}
&\tilde{\Phi}^{[n]}_{i,k}(\mathbf{w}_k,\xi_{i,k})
:=-2\mathrm{Re}\Big(\frac{(\mathbf{w}^{[n]}_k)^H\mathbf{h}_{i,k}\mathbf{h}^H_{i,k}\mathbf{w}_k}{\xi^{[n]}_{i,k}}\Big)
\nonumber\\
&~~~~~~~~~~~~~~~~~~~~
+
\frac{(\mathbf{w}^{[n]}_k)^H\mathbf{h}_{i,k}\mathbf{h}^H_{i,k}\mathbf{w}^{[n]}_k}{(\xi^{[n]}_{i,k})^2}\xi_{i,k},
\end{align}
and the following property can be established.
\begin{property}
The function $\tilde{\Phi}^{[n]}_{i,k}$ satisfies the following:
(i) $\tilde{\Phi}^{[n]}_{i,k}(\mathbf{w}_k,\xi_{i,k})\geq \Phi_{i,k}(\mathbf{w}_k,\xi_{i,k})$,
(ii) $\tilde{\Phi}^{[n]}_{i,k}(\mathbf{w}^{[n]}_k,\xi^{[n]}_{i,k})=\Phi_{i,k}(\mathbf{w}^{[n]}_k,\xi^{[n]}_{i,k})$, and
(iii)
\begin{align}
&\dfrac{\partial \tilde{\Phi}^{[n]}_{i,k}(\mathbf{w}_k,\xi_{i,k})}{\partial\mathbf{w}_k}
\Big|_{\mathbf{w}_k=\mathbf{w}^{[n]}_k,\xi_{i,k}=\xi^{[n]}_{i,k}}
\nonumber\\
&~~~~~~~~~~~~~~~~~~~~
=
\dfrac{\partial \Phi_{i,k}(\mathbf{w}_k,\xi_{i,k})}{\partial\mathbf{w}_k}
\Big|_{\mathbf{w}_k=\mathbf{w}^{[n]}_k,\xi_{i,k}=\xi^{[n]}_{i,k}}
,\nonumber
\\
&\dfrac{\partial \tilde{\Phi}^{[n]}_{i,k}(\mathbf{w}_k,\xi_{i,k})}{\partial\xi_{i,k}}
\Big|_{\mathbf{w}_k=\mathbf{w}^{[n]}_k,\xi_{i,k}=\xi^{[n]}_{i,k}}
\nonumber\\
&~~~~~~~~~~~~~~~~~~~~
=
\dfrac{\partial \Phi_{i,k}(\mathbf{w}_k,\xi_{i,k})}{\partial\xi_{i,k}}
\Big|_{\mathbf{w}_k=\mathbf{w}^{[n]}_k,\xi_{i,k}=\xi^{[n]}_{i,k}}.\nonumber
\end{align}
\end{property}
\begin{proof}
See Appendix E.
\end{proof}
In addition to being linear in $\mathbf{w}_k$ and $\xi_{i,k}$, from \textbf{Property 3}, we can see
that $\tilde{\Phi}^{[n]}_{i,k}$ has the same value and gradient as $\Phi_{i,k}$ at point $\{\mathbf{w}^{[n]}_k,\xi^{[n]}_{i,k}\}$.
More importantly, the first part of \textbf{Property 3} indicates that if we replace $\Phi_{i,k}$ with $\tilde{\Phi}^{[n]}_{i,k}$, the
feasible set becomes smaller, thus the solution of such a problem would be a feasible solution to the rank-relaxed problem of $\mathcal{P}2$.
Following a similar proof to \textbf{Property 3}, the term $-\mu^2_{i,k}$ in constraint \eqref{P4-mu} can be
conservatively replaced by $-2\mu^{[n]}_{i,k}\mu_{i,k}+(\mu^{[n]}_{i,k})^2$, where $\mu^{[n]}_{i,k}$ is the solution at the $n^{\mathrm{th}}$ iteration.
With the above observation, the following problem is considered at the $(n+1)^{\mathrm{th}}$ iteration
\begin{subequations}
\begin{align}
&\mathcal{P}2-R[n+1]:\nonumber\\
&
~~~~~\mathop{\mathrm{min}}_{\substack{\{\mathbf{V}_k,\mathbf{w}_k,\xi_{i,k},\beta_{i,k},\\\mu_{i,k},I_{j,l,k}\}}}
~\sum_{k=1}^K \mathrm{Tr}(\mathbf{V}_{k})+\sum_{i=1}^2\sum_{k=1}^K\frac{1}{\xi_{i,k}} \nonumber \\
&\mathrm{s.t.}~~
\alpha_{i,k}\Big(\sum_{j=1}^2\sum_{l\neq k}I_{j,l,k}+\sigma_r^2\Big)+\tilde{\Phi}^{[n]}_{i,k}
\leq 0,~~\forall i,k\nonumber
\\
&~~~~~\left[
\begin{array}{cccc}
I_{j,l,k} & \mathbf{w}^H_k\mathbf{h}_{j,l}
\\
\mathbf{h}^H_{j,l}\mathbf{w}_k & \xi_{j,l}
\end{array}
\right]
\succeq 0,~~\forall j\in\{1,2\},\forall k,l\in\mathcal{K},l\neq k,
\nonumber
\\
&~~~~~
\left[
\begin{array}{cccc}
\dfrac{\mathrm{Tr}(\mathbf{\Theta}_{i,k}\mathbf{V}_k)}{\theta_{i,k}}-\mathop{\sum}_{l\neq k}\mathrm{Tr}(\mathbf{\Theta}_{i,k}\mathbf{V}_{l})
-\sigma_u^2 & \sigma_z
\\
\sigma_z & \beta_{i,k}
\end{array}
\right]
\succeq 0,
\nonumber\\
&~~~~~~
\forall i,k\nonumber
\\
&~~~~~
\left[
\begin{array}{cccc}
\sum_{l=1}^K\mathrm{Tr}(\mathbf{\Theta}_{i,k}\mathbf{V}_{l})+\sigma_u^2 & \mu_{i,k}
\\
\mu_{i,k} & \eta(1-\beta_{i,k})
\end{array}
\right]
\succeq 0,~~\forall i,k \nonumber
\\
&~~~~~~\frac{1}{\xi_{i,k}}-2\mu^{[n]}_{i,k}\mu_{i,k}+(\mu^{[n]}_{i,k})^2+2p_c-2E_{i,k}\leq 0,
~~
\forall i,k \nonumber
\\
&~~~~~~
\mathbf{V}_k\succeq 0,~||\mathbf{w}_k||\leq1,~~\forall k. \nonumber
\end{align}
\end{subequations}
Now the problem $\mathcal{P}2-R[n+1]$ is an SDP problem which can be efficiently solved by CVX, a Matlab-based software package for convex optimization \cite{22}.
Denote its optimal solution as $\{\mathbf{V}^*_k,\mathbf{w}^*_k,\beta^*_{i,k},\xi^*_{i,k},\mu^*_{i,k}\}$.
Then we can set $\{\mathbf{w}^{[n+1]}_k=\mathbf{w}^*_k,\xi^{[n+1]}_{i,k}=\xi^*_{i,k},\mu^{[n+1]}_{i,k}=\mu^*_{i,k}\}$,
and the process repeats with solving the problem $\mathcal{P}2-R[n+2]$.
According to \textbf{Property 3} and \cite[Theorem 1]{25}, the iterative algorithm is convergent.
Furthermore, the converged point $\{\widetilde{\mathbf{V}}_k,\widetilde{\beta}_{i,k},\widetilde{\mathbf{w}}_k,\widetilde{\xi}_{i,k}\}$ is guaranteed to be a stationary point for the rank-relaxed problem of $\mathcal{P}2$\cite{25,A1,A6,B1,B2}.

In terms of computation effort, solving problem $\mathcal{P}2-R[n+1]$ is dominated by $K$ $N$-dimension SDP cones,
which requires complexity $O\Big(K^{3.5}N^{2.5}+K^{2.5}N^{3.5}\Big)$ \cite{20}.
Therefore, the total complexity is $O\Big(t(K^{3.5}N^{2.5}+K^{2.5}N^{3.5})\Big)$, where $t$ is the number of iterations needed for the inner approximation to converge, which is very small as shown by simulations.

With the obtained solution, the next stage is to recover the beamformer $\widetilde{\mathbf{v}}_k$ of $\mathcal{P}1$.
More specifically,
when $\widetilde{\mathbf{V}}_k$ is rank-one, the beamformer $\widetilde{\mathbf{v}}_k$ can be recovered by the rank-one decomposition of $\widetilde{\mathbf{V}}_k$.
On the other hand, when $\widetilde{\mathbf{V}}_k$ is rank-two,
by allocating power $\widetilde{p}_{k}=\mathrm{Tr}(\widetilde{\mathbf{V}}_k)$ at relay for
 the $k^{\mathrm{th}}$ downlink data stream,
 and applying rank-two decomposition, we obtain
 $\widetilde{\mathbf{V}}_k/\widetilde{p}_{k}=[\widetilde{\mathbf{f}}_{1,k}~\widetilde{\mathbf{f}}_{2,k}]
 [\widetilde{\mathbf{f}}_{1,k}~\widetilde{\mathbf{f}}_{2,k}]^H$.
 Then, the relay transmits $[\widetilde{\mathbf{f}}_{1,k}~\widetilde{\mathbf{f}}_{2,k}]\mathbf{S}_{k}(m)$ at each time duration for two symbols,
 where $\mathbf{S}_{k}(m)\in\mathbb{C}^{2\times 2}$ is the $k^{\mathrm{th}}$ downlink data stream grouped into blocks of Alamouti code:
\begin{equation}\label{35}
\mathbf{S}_{k}(m)=
\left[
\begin{array}{cccc}
s_{k}(2m-1) & s_{k}(2m)
\\
-s_{k}^*(2m) & s_{k}^*(2m-1)
\end{array}
\right], \nonumber
\end{equation}
with $m=1,...,M/2$.
As shown in \cite{28}, such a transmission incurs no loss when used together with a rank-two beamformer.
Consequently, the proposed algorithm is guaranteed to achieve at least a stationary point of $\mathcal{P}1$.

Finally, for the inner approximation algorithm, finding a good feasible starting point is of large importance \cite{B1}.
For $N\geq 2K-1$, since the degree of freedom at relay is sufficient to support the number of users,
$\mathcal{P}2$ is always feasible, and
a simple initialization is fixing the transmit and receive beamformer using pairwise zero-forcing criterion \cite[Section IV]{14}.
Then $\mathcal{P}2$ is an SDP with low dimension.
However, when $N<2K-1$, since the first constraint of $\mathcal{P}2$ involves DC functions, it is generally hard to find a feasible initial point or provide a feasibility condition \cite{A1,A6,B1}.
A traditional method is using $l_1$-norm regularization to pursuit a feasible point \cite{A6,B1}, which involves multiple conic programs (CPs).
Below, by exploiting the structure of problem $\mathcal{P}2$, we propose a CP-free initialization.

In particular, from the first two constraints of problem $\mathcal{P}2$, an initial point
$\{\mathbf{w}^{[0]}_k,\xi^{[0]}_{i,k}\}$ is feasible if and only if
\begin{align}
&
\mathop{\mathrm{min}}_{\substack{i\in\{1,2\}\\k\in\mathcal{K}}}\frac{|(\mathbf{w}^{[0]}_k)^H\mathbf{h}_{i,k}|^2/\xi^{[0]}_{i,k}}{\alpha_{i,k}\Big(\sum_{j=1}^2\sum_{l\neq k}|(\mathbf{w}^{[0]}_k)^H\mathbf{h}_{j,l}|^2/\xi^{[0]}_{j,l}+\sigma_r^2\Big)}\geq 1.
\nonumber
\end{align}
To find such a feasible solution, we consider the following problem with the objective function equal to the left hand side of the above inequality
\begin{align}\label{C1}
&\mathop{\mathrm{max}}_{\{\mathbf{w}_k,\xi_{i,k}>0\}}~
\mathop{\mathrm{min}}_{\substack{i\in\{1,2\}\\k\in\mathcal{K}}}\frac{|\mathbf{w}_k^H\mathbf{h}_{i,k}|^2/\xi_{i,k}}{\alpha_{i,k}\Big(\sum_{j=1}^2\sum_{l\neq k}|\mathbf{w}_k^H\mathbf{h}_{j,l}|^2/\xi_{j,l}+\sigma_r^2\Big)}
\nonumber\\
&~~~~~~\mathrm{s.t.}~~~~~~\sum_{i=1}^2\sum_{k=1}^K\frac{1}{\xi_{i,k}}\leq P,~||\mathbf{w}_k||=1,~~\forall k,
\end{align}
where $P$ is a sufficiently large control parameter.
Then a feasible initial point can be found by alternatively updating $\mathbf{w}_{k}$ and $\xi_{i,k}$ in \eqref{C1} until its objective function is greater than or equal to 1.
An advantage of such a \emph{block coordinate descent method} is that closed-form optimal solutions for each subproblem of \eqref{C1} can be obtained as proved in Appendix F.
After obtaining $\{\mathbf{w}^{[0]}_k,\xi^{[0]}_{i,k}\}$, we set $\mu^{[0]}_{i,k}=\sqrt{(1/\xi^{[0]}_{i,k}+2p_c-2E_{i,k})^+}$
according to the constraint \eqref{P4-mu} of $\mathcal{P}2$.
As each step is in a closed-form, the complexity of the proposed initialization is inconsequential.
Therefore, the proposed CP-free initialization also represents a \emph{low-complexity} solution for solving $\mathcal{P}2$.

Notice that even if $\mathcal{P}2$ is feasible, currently there is no method guaranteed to find a feasible solution due to the nature of DC programming problems \cite{B1}.
While the proposed CP-free method cannot guarantee a feasible initialization either, it has a very high probability of finding a feasible initialization since the block coordinate descent method converges to at least a stationary point of \eqref{C1} \cite[Section 3.1]{R5}.
Moreover, the probability of finding a feasible initialization can be further increased using multiple starters of $\{\xi_{i,k}\}$ \cite{A6}.

\section{Lower Bound and Large $N$ Analysis of $\mathcal{P}1$}
\subsection{Lower Bound of $\mathcal{P}1$}
In order to assess the performance of the proposed solution in Section IV-B, in this section, a design which provides a lower bound to the transmit power of
the optimal solution of $\mathcal{P}1$ is derived as a benchmark.
In particular, consider the following problem $\mathcal{P}1-R$:
\begin{subequations}
\begin{align}
&\mathcal{P}1-R:~\mathop{\mathrm{min}}_{\substack{\{\mathbf{V}_k,\mathbf{b}_{i,k},q_{i,k},\beta_{i,k}\}}}
~\sum_{k=1}^K \mathrm{Tr}(\mathbf{V}_k)+\sum_{i=1}^2\sum_{k=1}^K q_{i,k} \nonumber \\
&\mathrm{s.t.}~~q_{i,k}|\mathbf{b}^H_{i,k}\mathbf{h}_{i,k}|^2\geq\alpha_{i,k}
\Big(\sum_{j=1}^2\sum_{l\neq k}q_{j,l}|\mathbf{b}^H_{i,k}\mathbf{h}_{j,l}|^2
\nonumber\\
&~~~~~~~~~~~~~~~~~~~~~~~~~~~~~~~~~
+\sigma_r^2||\mathbf{b}_{i,k}||^2\Big),~~\forall i,k
\label{2ra}
\\
&~~~~~~
\beta_{i,k}\Big(\frac{\mathrm{Tr}(\mathbf{\Theta}_{i,k}\mathbf{V}_k)}{\theta_{i,k}}-\mathop{\sum}_{l\neq k}\mathrm{Tr}(\mathbf{\Theta}_{i,k}\mathbf{V}_{l})-\sigma_u^2\Big)
\nonumber\\
&~~~~~
\geq\sigma^2_z
,~~\forall i,k
\label{2rb}
\\
&~~~~~~
\eta(1-\beta_{i,k})(\sum_{l=1}^K\mathrm{Tr}(\mathbf{\Theta}_{i,k}\mathbf{V}_{l})+\sigma_u^2)
\nonumber\\
&~~~~~
\geq
q_{i,k}+2p_c-2E_{i,k}
,~~\forall i,k
\label{2rc}
\\
&~~~~~~q_{i,k}\geq0,~\beta_{i,k}\in (0,1),~~\forall i,k,~~
\mathbf{V}_k\succeq 0,~~\forall k.
\label{2rd}
\end{align}
\end{subequations}
It can be seen that if we restrict $\mathbf{b}_{1,k}=\mathbf{b}_{2,k}=\mathbf{w}_k$ and $\mathrm{Rank}(\mathbf{V}_{k})=1$ in $\mathcal{P}1-R$, it reduces to $\mathcal{P}1$.
Therefore, the feasible set of $\mathcal{P}1-R$ is larger than that of $\mathcal{P}1$, and $\mathcal{P}1-R$ is a relaxed version of $\mathcal{P}1$.
The insight behind such a relaxation is that $\{\mathbf{b}_{i,k}\}$ represents \emph{virtual receivers} at relay, which attenuates
the pairwise uplink-downlink coupling.
In the following,
we will show that $\mathcal{P}1-R$ can be optimally solved in two steps.

From $\mathcal{P}1-R$, it can be seen that $\{\mathbf{b}_{i,k},q_{i,k}\}$ is only involved in \eqref{2ra}, \eqref{2rc} and the first part of \eqref{2rd}.
Observing from \eqref{2rc} that smaller $q_{i,k}$ loosens the constraints on $\mathbf{V}_k$, which helps in reducing the objective value, the optimal $\{q^*_{i,k},\mathbf{b}^*_{i,k}\}$ of $\mathcal{P}1-R$ can be therefore obtained from solving
\begin{align}\label{P5R}
&\mathcal{P}5:\mathop{\mathrm{min}}_{\{\mathbf{b}_{i,k},q_{i,k}\geq0\}}~\Big[q_{1,1},q_{1,2},...,q_{1,K},q_{2,K}\Big]
\nonumber\\
&~\mathrm{s.t.}~~~q_{i,k}|\mathbf{b}^H_{i,k}\mathbf{h}_{i,k}|^2\geq\alpha_{i,k}
\Big(\sum_{j=1}^2\sum_{l\neq k}q_{j,l}|\mathbf{b}^H_{i,k}\mathbf{h}_{j,l}|^2
\nonumber\\
&~~~~~~~~~~~~~~~~~~~~~~~~~~~~~~~~~~~
+\sigma_r^2||\mathbf{b}_{i,k}||^2\Big),~~\forall i,k.
\end{align}
Then by putting the optimal $\{q^*_{i,k},\mathbf{b}^*_{i,k}\}$ of $\mathcal{P}5$ into $\mathcal{P}1-R$, the
problem can be simplified and the optimal $\{\mathbf{V}^*_k,\beta^*_{i,k}\}$ can be found from the simplified problem.

To solve problem $\mathcal{P}5$, which is a multi-criterion optimization,
consider the following iteration:
\begin{align}
&\mathbf{z}^{[n+1]}_{i,k}=\frac{\Big(\sum_{j}\sum_{l\neq k}\omega^{[n]}_{j,l}\mathbf{h}_{j,l}\mathbf{h}^H_{j,l}+\sigma_r^2\mathbf{I}\Big)^{-1}\mathbf{h}_{i,k}}{\Big|\Big|\Big(\sum_{j}\sum_{l\neq k}\omega^{[n]}_{j,l}\mathbf{h}_{j,l}\mathbf{h}^H_{j,l}+\sigma_r^2\mathbf{I}\Big)^{-1}\mathbf{h}_{i,k}\Big|\Big|},
\label{bik}
\\
&\omega^{[n+1]}_{i,k}=\frac{\alpha_{i,k}(\sum_{j=1}^2\sum_{l\neq k}\omega^{[n]}_{j,l}|(\mathbf{z}^{[n+1]}_{i,k})^H\mathbf{h}_{j,l}|^2+\sigma_r^2)}{|(\mathbf{z}^{[n+1]}_{i,k})^H\mathbf{h}_{i,k}|^2},
\label{qik}
\end{align}
and the following proposition can be established.
\begin{proposition}
With $\omega_{i,k}^{[0]}=0$ for all $i,k$, the sequence $\omega^{[n]}_{i,k}$ is convergent, and the limit point is the optimal solution of $\mathcal{P}5$.
\end{proposition}
\begin{proof}
First we prove $\omega^{[n]}_{i,k}$ is monotonically increasing by induction.
Since $\omega_{i,k}^{[0]}=0\leq \omega_{i,k}^{[1]}$, we assume that $\omega_{i,k}^{[n-1]}\leq \omega_{i,k}^{[n]}$ for all $(i,k)$ with some $n\geq1$. From \eqref{qik}, we have
\begin{align}
&\omega^{[n]}_{i,k}=\frac{\alpha_{i,k}(\sum_{j=1}^2\sum_{l\neq k}\omega^{[n-1]}_{j,l}|(\mathbf{z}^{[n]}_{i,k})^H\mathbf{h}_{j,l}|^2+\sigma_r^2)}{|(\mathbf{z}^{[n]}_{i,k})^H\mathbf{h}_{i,k}|^2}
\nonumber\\
&~~~~\leq
\frac{\alpha_{i,k}(\sum_{j=1}^2\sum_{l\neq k}\omega^{[n-1]}_{j,l}|(\mathbf{z}^{[n+1]}_{i,k})^H\mathbf{h}_{j,l}|^2+\sigma_r^2)}{|(\mathbf{z}^{[n+1]}_{i,k})^H\mathbf{h}_{i,k}|^2}
\nonumber\\
&~~~~\leq
\frac{\alpha_{i,k}(\sum_{j=1}^2\sum_{l\neq k}\omega^{[n]}_{j,l}|(\mathbf{z}^{[n+1]}_{i,k})^H\mathbf{h}_{j,l}|^2+\sigma_r^2)}{|(\mathbf{z}^{[n+1]}_{i,k})^H\mathbf{h}_{i,k}|^2}
\nonumber\\
&~~~~=
\omega^{[n+1]}_{i,k}
, \label{i}
\end{align}
where the inequality in the second line is due to
\begin{align}
\mathbf{z}^{[n]}_{i,k}=\mathrm{arg}~\mathop{\mathrm{min}}_{\mathbf{z}_{i,k}}~\frac{\alpha_{i,k}(\sum_{j=1}^2\sum_{l\neq k}\omega^{[n-1]}_{j,l}|\mathbf{z}_{i,k}^H\mathbf{h}_{j,l}|^2+\sigma_r^2)}{|\mathbf{z}_{i,k}^H\mathbf{h}_{i,k}|^2}
\label{reasonbik}
\end{align}
which can be obtained from \eqref{bik} and \cite[Corollary 2]{27},
the inequality in the third line is due to $\omega_{i,k}^{[n-1]}\leq \omega_{i,k}^{[n]}$,
and the equality in the last line is due to \eqref{qik}.
Thus $\omega^{[n]}_{i,k}$ is monotonically increasing.

Then we prove $\omega^{[n]}_{i,k}$ is upper bounded.
Specifically, we will show that
$\omega_{i,k}^{[n]}\leq \overline{q}_{i,k}$ by induction, where
$\{\overline{q}_{i,k},\mathbf{\overline{b}}_{i,k}\}$ is an arbitrary feasible solution of $\mathcal{P}5$.
Since $\omega_{i,k}^{[0]}=0\leq \overline{q}_{i,k}$, we assume that $\omega_{i,k}^{[n]}\leq \overline{q}_{i,k}$ for some $n\geq0$.
Using \eqref{qik}, we have
\begin{align}
&\omega^{[n+1]}_{i,k}=\frac{\alpha_{i,k}(\sum_{j=1}^2\sum_{l\neq k}\omega^{[n]}_{j,l}|(\mathbf{z}^{[n+1]}_{i,k})^H\mathbf{h}_{j,l}|^2+\sigma_r^2)}{|(\mathbf{z}^{[n+1]}_{i,k})^H\mathbf{h}_{i,k}|^2}
\nonumber\\
&~~~~~~~\leq
\frac{\alpha_{i,k}(\sum_{j=1}^2\sum_{l\neq k}\omega^{[n]}_{j,l}|\mathbf{\overline{b}}_{i,k}^H\mathbf{h}_{j,l}|^2+\sigma_r^2)}{|\mathbf{\overline{b}}_{i,k}^H\mathbf{h}_{i,k}|^2}
\nonumber\\
&~~~~~~~\leq
\frac{\alpha_{i,k}(\sum_{j=1}^2\sum_{l\neq k}\overline{q}_{j,l}|\mathbf{\overline{b}}_{i,k}^H\mathbf{h}_{j,l}|^2+\sigma_r^2)}{|\mathbf{\overline{b}}_{i,k}^H\mathbf{h}_{i,k}|^2}
\nonumber\\
&~~~~~~~\leq
\overline{q}_{i,k}
,
\end{align}
where the inequality in the second line is due to \eqref{reasonbik},
the inequality in the third line is due to $\omega_{i,k}^{[n]}\leq \overline{q}_{i,k}$,
and the inequality in the last line is due to the constraint in \eqref{P5R}.
Therefore as long as $\mathcal{P}5$ has a feasible solution, the sequence $\omega_{i,k}^{[n]}$ is upper bounded.

Lastly, using the above two results,
we immediately have $\omega_{i,k}^{[n]}$ is convergent \cite{19}.
Furthermore, by putting $\omega_{i,k}^{[n+1]}=\omega_{i,k}^{[n]}=\omega_{i,k}^{[\infty]}$ into \eqref{qik}, the limit point $\omega_{i,k}^{[\infty]}$ satisfies
\begin{align}
&\omega^{[\infty]}_{i,k}=\frac{\alpha_{i,k}(\sum_{j=1}^2\sum_{l\neq k}\omega^{[\infty]}_{j,l}|(\mathbf{z}^{[\infty]}_{i,k})^H\mathbf{h}_{j,l}|^2+\sigma_r^2)}{|(\mathbf{z}^{[\infty]}_{i,k})^H\mathbf{h}_{i,k}|^2},
\end{align}
which implies that $\{\omega^{[\infty]}_{i,k},\mathbf{z}^{[\infty]}_{i,k}\}$ is a feasible solution for $\mathcal{P}5$.
On the other hand, since we obtained $\omega_{i,k}^{[n]}\leq \overline{q}_{i,k}$ for any feasible $\overline{q}_{i,k}$,
$\{\omega^{[\infty]}_{i,k},\mathbf{z}^{[\infty]}_{i,k}\}$ is therefore optimal for $\mathcal{P}5$.
\end{proof}

Based on \textbf{Proposition 3}, the optimal solution of $\mathcal{P}5$ is $\{q_{i,k}=\omega_{i,k}^{[\infty]},\mathbf{b}_{i,k}=\mathbf{z}_{i,k}^{[\infty]}\}$.
Then by substituting them into $\mathcal{P}1-R$, the problem $\mathcal{P}1-R$ reduces to
\begin{align}
&~~~~~\mathop{\mathrm{min}}_{\substack{\{\mathbf{V}_k\succeq 0,\beta_{i,k}\}}}
~\sum_{k=1}^K\mathrm{Tr}(\mathbf{V}_{k})+\sum_{i=1}^2\sum_{k=1}^K\omega^{[\infty]}_{i,k} \nonumber\\
&\mathrm{s.t.}~\left[
\begin{array}{cccc}
\dfrac{\mathrm{Tr}(\mathbf{\Theta}_{i,k}\mathbf{V}_k)}{\theta_{i,k}}-\mathop{\sum}_{l\neq k}\mathrm{Tr}(\mathbf{\Theta}_{i,k}\mathbf{V}_{l})
-\sigma_u^2 & \sigma_z
\\
\sigma_z & \beta_{i,k}
\end{array}
\right]
\succeq 0,
\nonumber\\
&~~~~~~\forall i,k
\nonumber\\
&
\left[
\begin{array}{cccc}
\sum_{l=1}^K\mathrm{Tr}(\mathbf{\Theta}_{i,k}\mathbf{V}_{l})+\sigma_u^2 & \sqrt{(\omega^{[\infty]}_{i,k}+2p_c-2E_{i,k})^+}
\\
\sqrt{(\omega^{[\infty]}_{i,k}+2p_c-2E_{i,k})^+} & \eta(1-\beta_{i,k})
\end{array}
\right]
\nonumber\\
&
\succeq 0,~~\forall i,k.
\nonumber
\end{align}
The above problem is an SDP which can be optimally solved by CVX.
Since we have shown that $\mathcal{P}1-R$ can be equivalently transformed into a two-step optimization, and we solve each step optimally,
the obtained solution is optimal for $\mathcal{P}1-R$, and would be a lower bound for $\mathcal{P}1$.

\subsection{Large $N$ Analysis of $\mathcal{P}1$}
In this subsection, wirelessly powered multi-pair TWR with very large number of antennas at relay, i.e., $N\rightarrow\infty$, is analyzed.
Specifically, when $N$ is very large, user channels will be asymptotically orthogonal, and the interference vanishes \cite[Proposition 1]{R2}.
Therefore, problem $\mathcal{P}1$ asymptotically becomes
\begin{align}
&\mathcal{P}6:\mathop{\mathrm{min}}_{\substack{\{\mathbf{v}_{i,k},\mathbf{w}_{i,k},p_{i,k},\\q_{i,k},\beta_{i,k}\}}}
~\sum_{i=1}^2\sum_{k=1}^K\Big(p_{i,k}+q_{i,k}\Big) \nonumber\\
&\mathrm{s.t.}~~
\frac{q_{i,k}|\mathbf{w}^H_{i,k}\mathbf{h}_{i,k}|^2}{\sigma_r^2}
\geq \theta_{3-i,k},~~\forall i,k
\nonumber
\\
&~~~~~~
\frac{\beta_{i,k}p_{i,k}|\mathbf{g}^H_{i,k}\mathbf{v}_{i,k}|^2}{\beta_{i,k}\sigma_u^2+\sigma_z^2}
\geq \theta_{i,k},~~
\forall i,k
\nonumber \\
&~~~~~~\eta(1-\beta_{i,k})\Big(p_{i,k}|\mathbf{g}^H_{i,k}\mathbf{v}_{i,k}|^2+\sigma_u^2\Big)+2E_{i,k}-2p_c\geq q_{i,k},
\nonumber\\
&~~~~~~
\forall i,k
\nonumber\\
&
~~~~~~q_{i,k}\geq0,~\beta_{i,k}\in (0,1],~~\forall i,k
\nonumber\\
&~~~~~~||\mathbf{v}_{i,k}||=1,~||\mathbf{w}_{i,k}||=1,~p_{i,k}\geq0,~~\forall i,k.
\nonumber
\end{align}
Notice that there is no need to apply network coding since all the channels are asymptotically orthogonal.
In this case, the uplink and downlink SINR targets become $\theta_{3-i,k}=2^{2\overline{R}_{i,k}}-1$
and $\theta_{i,k}=2^{2\overline{R}_{3-i,k}}-1$, respectively.
Observing from $\mathcal{P}6$ that minimizing $p_{i,k}$ and $q_{i,k}$ requires matching the beamformers to their corresponding channels, the optimal $\mathbf{w}^*_{i,k}=\mathbf{h}_{i,k}/||\mathbf{h}_{i,k}||$ and
$\mathbf{v}^*_{i,k}=\mathbf{g}_{i,k}/||\mathbf{g}_{i,k}||$.
Putting $\mathbf{w}^*_{i,k}$ into the first constraint of $\mathcal{P}6$, it reduces into
$q_{i,k}||\mathbf{h}_{i,k}||^2
\geq \theta_{3-i,k}\sigma_r^2$.
Since the effect of fast fading vanishes when $N$ is large \cite{R2}, we have $||\mathbf{h}_{i,k}||^2=N\varrho_{i,k}$, where
$\varrho_{i,k}$ is the large-scale fading of $\mathbf{h}_{i,k}$,
and thus
\begin{align}
q^*_{i,k}=\frac{\theta_{3-i,k}\sigma^2_r}{N\varrho_{i,k}}.\nonumber
\end{align}
This result is consistent with the fact that information beamforming gain is proportional to $N$ in large-scale antenna systems \cite{R2}.
Therefore, given certain data-rate QoS, the users' transmit powers can be decreased proportionally to $1/N$ while still satisfying the QoS constraints.

To determine the network transmit power, we also need to solve for $p^*_{i,k}$ in $\mathcal{P}6$.
Based on the procedure given in Appendix G, its optimal solution is given by
\begin{align}\label{pik}
&p^*_{i,k}=
\left\{
\begin{aligned}
&\frac{\theta_{i,k}}{N\varrho_{i,k}}(\sigma_u^2+\sigma_z^2),~~~~~~~~~
\mathrm{if}~\frac{\theta_{3-i,k}\sigma^2_r}{N\varrho_{i,k}}+2p_c\leq2E_{i,k}
\\
&
\frac{\theta_{i,k}}{N\varrho_{i,k}}\Big(
\frac
{B_{i,k}
+\sqrt{B^2_{i,k}+4\theta_{i,k}(\theta_{i,k}+1)\sigma^2_u\sigma^2_z}}
{2\theta_{i,k}}
\\
&~~~~~~~~~+\sigma_u^2
\Big),~~~~~~~~~~~
\mathrm{if}~\frac{\theta_{3-i,k}\sigma^2_r}{N\varrho_{i,k}}+2p_c>2E_{i,k}
\end{aligned}
\right.,
\end{align}
with $B_{i,k}=\theta_{i,k}\sigma^2_z-(\theta_{i,k}+1)\sigma^2_u+[\theta_{3-i,k}\sigma^2_r(N\varrho_{i,k})^{-1}+2p_c-2E_{i,k}]/\eta$.
From the analytical solution of $p^*_{i,k}$, we can see that
the relay transmit power can be decreased proportionally to $1/N$ as well.
This result is consistent with the fact that energy beamforming gain is proportional to $N$ in large-scale antenna systems \cite{R3}.
Based on the scaling law of relay and user transmit powers, we can conclude that when $N$ is large, the network transmit power can be reduced proportionally to $1/N$ while still satisfying the data-rate QoS constraints and energy harvesting constraints.

\section{Simulation Results and Discussions}

This section provides simulation results to demonstrate the performance of the proposed schemes.
In particular,
the distance-dependent pathloss model of the $(i,k)^{\mathrm{th}}$ user $\varrho_{i,k}=\varrho_0\cdot(\frac{d_{i,k}}{d_0})^{-\alpha}$ is adopted \cite{11}, where $d_{i,k}$ is the distance from the $(i,k)^{\mathrm{th}}$ user to the relay, $\varrho_0=10^{-3}$, $d_0=1\mathrm{m}$ is the reference distance,
and $\alpha$ is the pathloss exponent set to be $2.7$ \cite{J3,8}.
In the simulations, $d_{i,k}\sim\mathcal{U}(1,10)$ in meter, where $\mathcal{U}$ represents the uniform distribution, and $\mathbf{g}_{i,k},\mathbf{h}_{i,k}$ are generated according to $\mathcal{CN}(\mathbf{0},\varrho_{i,k}\mathbf{I})$.
It is assumed that power conversion efficiency $\eta=0.8$ \cite{R1}, and noise power $\sigma^2_r=\sigma^2_u=\sigma^2_z=-60\mathrm{dBm}$ \cite{8}.
The circuit power consumption is set to $p_c=10\mathrm{dBm}$ \cite{C1}, and $E_{i,k}\sim \mathcal{U}(9.5,13.0)$ in $\mathrm{dBm}$.
For the simulations of transmit power versus data-rate QoS, the same data-rate targets $\overline{R}_{i,k}=R$ are requested by all users \cite{16}, while for other simulations $\overline{R}_{i,k}\sim \mathcal{U}(0,2)$ in bps/Hz.
Each point in the figures is obtained by averaging over 100 simulation runs, with independent channels in each run.

\begin{figure}[!t]
\centering
\includegraphics[width=75mm]{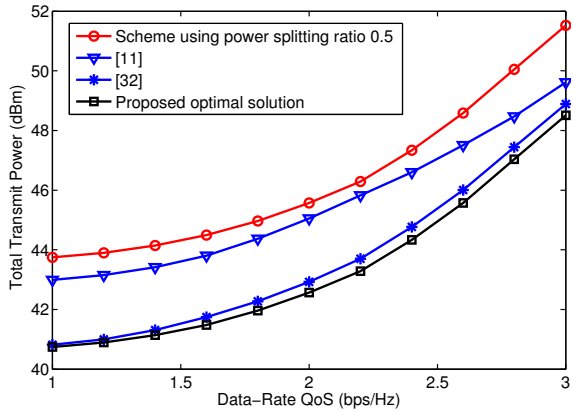}
\caption{Total transmit power versus data-rate QoS for the case of $K=1$ with $N=8$ when $E_{i,k}\sim \mathcal{U}(9.5,13.0)$ in dBm and noise power is $-60$dBm.}
\label{fig_sim}
\end{figure}

\begin{figure}[!t]
\centering
\includegraphics[width=75mm]{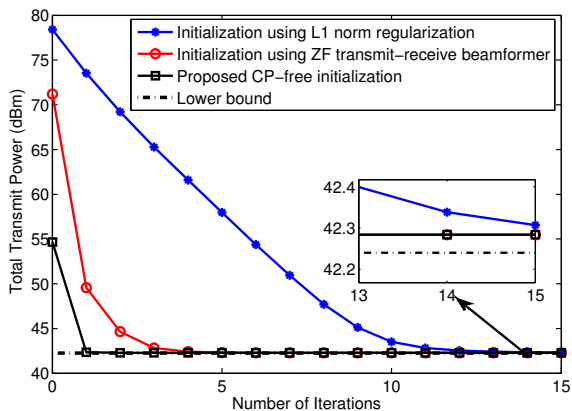}
\caption{Total transmit power versus number of iterations for the case of $K=3$ with $N=12$ when
$\overline{R}_{i,k}\sim \mathcal{U}(0,2)$ in bps/Hz, $E_{i,k}\sim \mathcal{U}(9.5,13.0)$ in dBm, and noise power is $-60$dBm.}
\label{fig_sim}
\end{figure}

First we consider the one pair case $K=1$ with $N=8$.
Here, four schemes are compared: the optimal solution, the SDP solution from \cite{R4}, the iterative solution from \cite{12}, and the solution with $\beta_1=\beta_2=0.5$.
As is shown in Fig. 2, the solution in Section IV-A of this paper achieves the lowest transmit power over a wide range of data-rate QoS.
Compared to the suboptimal solutions in \cite{R4} and \cite{12}, the proposed scheme has an average advantage of $0.3\mathrm{dB}$ and $2.1\mathrm{dB}$, respectively.
It is also observed that fixing power splitting ratio leads to more than $3\mathrm{dB}$ loss compared to the proposed optimal solution, indicating that the power splitting ratios should be jointly optimized with the beamformers.

Next we consider the case of three pairs of users $K=3$ with $N=12$.
To verify the convergence of the proposed iterative method in Section IV-B,
Fig. 3 shows the total transmit power versus number of iterations.
It can be seen that different initializations for the proposed iterative algorithm converge to the same value, which is within $0.2\mathrm{dB}$ from the lower bound.
Furthermore, among the initializations, the proposed CP-free initialization converges the fastest and stabilizes after $3$ iterations,
indicating that the complexity of the iterative algorithm can be moderate with the CP-free initialization.
For the $l_1$-norm regularization, due to its random picking of starting points, it requires more than $10$ iterations to converge.
For the ZF beamforming initialization, since it completely eliminates the inter-pair interference and forces the harvested energy at users to come from useful signals only, the initial transmit power is high.
But it also shows fast convergence, making it a compelling initialization when $N\geq2K-1$.
For the rest of this section, the proposed CP-free initialization will be used and the iterative algorithm stops after $3$ iterations.

\begin{figure}[!t]
\centering
\includegraphics[width=75mm]{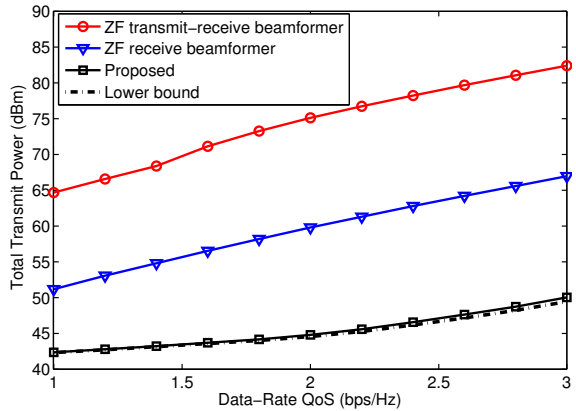}
\caption{Total transmit power versus data-rate QoS for the case of $K=3$ with $N=12$ when $E_{i,k}\sim \mathcal{U}(9.5,13.0)$ in dBm and noise power is $-60$dBm.}
\label{fig_sim}
\end{figure}

\begin{figure}[!t]
\centering
\includegraphics[width=75mm]{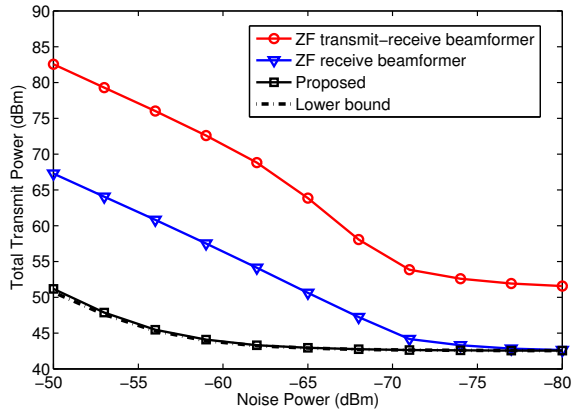}
\caption{Total transmit power versus noise power for the case of $K=3$ with $N=12$
when $\overline{R}_{i,k}\sim \mathcal{U}(0,2)$ in bps/Hz and $E_{i,k}\sim \mathcal{U}(9.5,13.0)$ in dBm.}
\label{fig_sim}
\end{figure}

To further demonstrate the performance of the proposed iterative algorithm for solving $\mathcal{P}2$,
Fig. 4 and Fig. 5 show the total transmit power versus the data-rate QoS and the noise power, respectively.
Apart from the ZF transmit-receive beamforming scheme, we also simulated the ZF receive beamforming scheme \cite[Section V]{11}, which optimizes the transmit beamformer in $\mathcal{P}2$ and provides more flexibility than the ZF transmit-receive beamforming scheme.
It can be observed from both figures that the proposed algorithm significantly outperforms other schemes and approaches the lower bound very tightly over a wide range of data-rate QoS and noise power.
This again verifies the excellent performance of the proposed method.
Interestingly, from Fig. 5, the performance of the suboptimal scheme using ZF receive beamformer also approaches the lower bound when the noise power is extremely small, e.g., smaller than $-75\mathrm{dBm}$.
As the scheme using ZF receive beamformer has a low complexity, it represents a promising solution in high SNR regime.

The aforesaid experiments examine the situation with $N\geq2K-1$.
Now we focus on the case of $N<2K-1$, in which the ZF schemes are invalid but the proposed algorithm is still applicable.
Specifically, the cases of $N=2K-2$ with $K=5$ and $K=7$ are simulated.
It can be seen from Fig. 6 that
although the gap between the proposed algorithm and the lower bound is larger due to the lower degree of freedom,
the proposed algorithm still achieves satisfying performance very close to the lower bound.
This reflects the high efficiency of the proposed method even in the case of low degree of freedom.

\begin{figure}[!t]
\centering
\includegraphics[width=75mm]{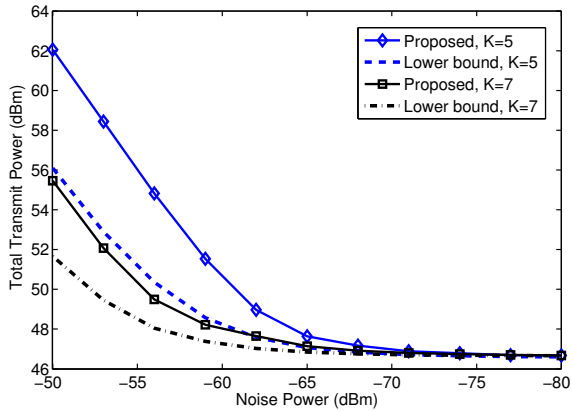}
\caption{Total transmit power versus noise power for the case of low degree of freedom when $\overline{R}_{i,k}\sim \mathcal{U}(0,2)$ in bps/Hz and $E_{i,k}\sim \mathcal{U}(9.5,13.0)$ in dBm.}
\label{fig_sim}
\end{figure}

\section{Conclusions}

This paper studied the total transmit power minimization in the multi-pair TWR system with harvest-then-transmit users.
Despite the challenge brought by the pairwise uplink-downlink coupling, it was proved that optimality could be achieved in the special case of one pair of users.
For the general case with multiple pairs of users, a convergence guaranteed iterative algorithm was proposed, and a lower bound on the performance of the optimal solution was derived.
Simulation results demonstrated that the proposed method outperforms other methods, and achieves good performance very close to the derived lower bound.

\appendices

\section{Compute-and-Forward with Lattice Codes}

\vspace{0.05in}
\noindent\underline{Generating $\mathbf{x}_{i,k}$ at the $(i,k)^{\mathrm{th}}$ user}
\vspace{0.05in}

Given $q_{i,k}$ and $|\mathbf{w}^H_k\mathbf{h}_{i,k}|^2$, we choose $K$ pairs of $M \times 1$ doubly nested lattice \cite{17} such that $\mathrm{\Lambda}_{1,k}\subseteq\mathrm{\Lambda}_{2,k}\subseteq\mathrm{\Lambda}_{k}$ with $k\in\mathcal{K}$, and the second moments
$\sigma^2(\Lambda_{i,k})=q_{i,k}|\mathbf{w}^H_{k}\mathbf{h}_{i,k}|^2$.
For user $(i,k)$, the source data is mapped into $\mathbf{c}_{i,k}\in L_{i,k}=\{\mathrm{\Lambda}_k~\mathrm{mod}~\mathrm{\Lambda}_{i,k}\}$ and the symbols to be transmitted are
\begin{eqnarray}\label{2}
&\mathbf{x}_{i,k}={(\mathbf{w}^H_k\mathbf{h}_{i,k})}^{-1}[(\mathbf{c}_{i,k}+\mathbf{d}_{i,k})\mathrm{mod}~\Lambda_{i,k}],
\end{eqnarray}
where $\mathbf{d}_{i,k}$ is a pre-generated random dither vector known to the $k^{\mathrm{th}}$ pair of users and relay.
Notice that the transmit power of $\mathbf{x}_{i,k}$ is $\frac{1}{M}\mathbb{E}[||\mathbf{x}_{i,k}||^2]=q_{i,k}$.

\vspace{0.05in}
\noindent\underline{Generating $\mathbf{s}_{k}$ at relay}
\vspace{0.05in}

Putting \eqref{2} into \eqref{1} and applying $\mathbf{w}_k$ to
extract the $k^{\mathrm{th}}$ layer signal at relay, we have
\begin{align}
&\mathbf{w}^H_{k}\mathbf{Y}=\Big[\mathop{\sum}_{i=1}^2(\mathbf{c}_{i,k}+\mathbf{d}_{i,k})\mathrm{mod}~\Lambda_{i,k}\Big]^T
\nonumber\\
&~~~~~~~~~~~
+\underbrace{\mathbf{w}^H_k\mathop{\sum}_{l\neq k}(\mathbf{h}_{1,l}\mathbf{x}^T_{1,l}+\mathbf{h}_{2,l}\mathbf{x}^T_{2,l})
+\mathbf{w}^H_k\mathbf{N}}_{\mathbf{m}^T_{k}},\nonumber
\end{align}
where the first term is the useful signal, the second term is inter-pair~interference, and the last term is noises with power $\frac{1}{M}\mathbb{E}[||\mathbf{w}^H_k\mathbf{N}||^2]=\sigma^2_r$.
Then, the relay computes \cite{17}:
\begin{align}
&\Big\{\kappa_k(\mathbf{w}^H_{k}\mathbf{Y})^T-\Sigma_{i=1}^2\mathbf{d}_{i,k}\Big\}\mathrm{mod}~\Lambda_{1,k}
\nonumber\\
&=\Big\{
\Big[\mathop{\sum}_{i=1}^2(\mathbf{c}_{i,k}+\mathbf{d}_{i,k})\mathrm{mod}~\Lambda_{i,k}\Big]
-\Sigma_{i=1}^2\mathbf{d}_{i,k}\nonumber\\
&~~~
\underbrace{-(1-\kappa_k)\Big[\mathop{\sum}_{i=1}^2(\mathbf{c}_{i,k}+\mathbf{d}_{i,k})\mathrm{mod}~\Lambda_{i,k}\Big]+\kappa_k\mathbf{m}_{k}}_{\tilde{\mathbf{m}}_k}
\Big\}\mathrm{mod}~\Lambda_{1,k}\nonumber\\
&=\Big\{
\underbrace{\mathop{\sum}_{i=1}^2\Big[\mathbf{c}_{i,k}-
Q_{i,k}(\mathbf{c}_{i,k}+\mathbf{d}_{i,k})\Big]}_{\mathbf{t}_k}
+\tilde{\mathbf{m}}_k\Big\}\mathrm{mod}~\Lambda_{1,k}
\label{tk}
\end{align}
where
$Q_{i,k}(\mathbf{c}_{i,k}+\mathbf{d}_{i,k})
=(\mathbf{c}_{i,k}+\mathbf{d}_{i,k})-(\mathbf{c}_{i,k}+\mathbf{d}_{i,k})\mathrm{mod}~\Lambda_{i,k}$
represents the nearest neighbor lattice of $\mathbf{c}_{i,k}+\mathbf{d}_{i,k}$ in $\Lambda_{i,k}$, and
\begin{align}
\kappa_k=\frac{q_{1,k}|\mathbf{w}^H_{k}\mathbf{h}_{1,k}|^2+q_{2,k}|\mathbf{w}^H_{k}\mathbf{h}_{2,k}|^2}
{\sum_{j=1}^2\sum_{l=1}^Kq_{j,l}|\mathbf{w}^H_k\mathbf{h}_{j,l}|^2+\sigma^2_r}.\nonumber
\end{align}
Since $\tilde{\mathbf{m}}_k$ is the suppressed interference (whitened by random dither, it can be viewed as Gaussian when $M\rightarrow\infty$) plus noise, the relay can recover $\mathbf{t}_k$ by applying lattice decoding to equation \eqref{tk}.
Then the relay maps the retrieved $\mathbf{t}_k$ to symbol $\mathbf{s}_{k}(\mathbf{t}_k)\in L_{r,k}$,
where $L_{r,k}$ is the lattice codebook at relay (not related to $L_{1,k}$ and $L_{2,k}$), with $\sigma^2(\Lambda_{r,k})=p_{k}$.

\section{Proof of Property 1}
To prove the property,
we first show that any stationary point of $\mathcal{P}1$ activates the first constraint for all users.
For simplicity, define the left hand side of the first constraint of $\mathcal{P}1$ as
\begin{align}
&\Delta_{i,k}:=\frac{q_{i,k}|\mathbf{w}^H_{k}\mathbf{h}_{i,k}|^2}
{\sum_{j=1}^2q_{j,k}|\mathbf{w}^H_{k}\mathbf{h}_{j,k}|^2}
\nonumber\\
&~~~~~~~~~~
+
\frac{q_{i,k}|\mathbf{w}^H_k\mathbf{h}_{i,k}|^2}{\sum_{j=1}^2\sum_{l\neq k}q_{j,l}|\mathbf{w}^H_k\mathbf{h}_{j,l}|^2+\sigma_r^2}\label{9}.
\end{align}
Then the Lagrangian of $\mathcal{P}1$ with respect to $\{q_{i,k}\}$ could be written as \cite{22}
\begin{align}
&\mathcal{L}=\sum_{i=1}^2\sum_{k=1}^K q_{i,k}
+\sum_{i=1}^2\sum_{k=1}^K\lambda_{i,k}(2^{2\overline{R}_{i,k}}-\Delta_{i,k})
\nonumber\\
&~~~~~
+\sum_{i=1}^2\sum_{k=1}^K\rho_{i,k}\Big(q_{i,k}-\eta(1-\beta_{i,k})(\sum_{l=1}^Kp_{l}|\mathbf{g}^H_{i,k}\mathbf{v}_{l}|^2+\sigma_u^2)
\nonumber\\
&~~~~~~~~~~~~~~~~~~~~~~~
+2p_c-2E_{i,k}\Big)
-\sum_{i=1}^2\sum_{k=1}^K\nu_{i,k}q_{i,k}
,\nonumber
\end{align}
with Lagrangian multipliers $\lambda_{i,k}\geq0,\rho_{i,k}\geq0,\nu_{i,k}\geq0$.
According to the KKT condition $\partial L/\partial q_{i,k}=0$ for all $i,k$, we have
\begin{align}
&1
-\sum_{j=1}^2\sum_{l=1}^K\lambda_{j,l}\frac{\partial\Delta_{j,l}}{\partial q_{i,k}}
+\rho_{i,k}
-\nu_{i,k}=0,~~\forall i,k.\label{kkt0}
\end{align}
Since $\overline{R}_{i,k}>0$, the user power satisfies $q_{i,k}\neq 0$.
From complementary slackness $\nu_{i,k}q_{i,k}=0$, equation $\nu_{i,k}=0$ holds.
Putting $\nu_{i,k}=0$ into \eqref{kkt0} gives
\begin{align}
&
\frac{\partial\Delta_{i,k}}{\partial q_{i,k}}
\lambda_{i,k}=
1
+\rho_{i,k}
-\sum_{(j,l)\neq(i,k)}\lambda_{j,l}\frac{\partial\Delta_{j,l}}{\partial q_{i,k}},~~\forall i,k. \label{kkt1}
\end{align}

To analyze the value of $\lambda_{i,k}$, we need to determine the signs of $\partial\Delta_{j,l}/\partial q_{i,k}$ for $(j,l)\neq(i,k)$.
More specifically, dividing the numerator and denominator of the first term in \eqref{9} by $q_{i,k}|\mathbf{w}^H_{k}\mathbf{h}_{i,k}|^2$, we have
\begin{align}
&\Delta_{i,k}:=1\Big/\Big(1+\frac{q_{3-i,k}|\mathbf{w}^H_{k}\mathbf{h}_{3-i,k}|^2}{q_{i,k}|\mathbf{w}^H_{k}\mathbf{h}_{i,k}|^2}\Big)
\nonumber\\
&~~~~~~~~~
+\frac{q_{i,k}|\mathbf{w}^H_k\mathbf{h}_{i,k}|^2}{\sum_{j=1}^2\sum_{l\neq k}q_{j,l}|\mathbf{w}^H_k\mathbf{h}_{j,l}|^2+\sigma_r^2}. \label{delta}
\end{align}
Clearly, for $\Delta_{j,l}$ with $(j,l)\neq(i,k)$,
when $q_{i,k}$ decreases,
$\Delta_{j,l}$ will monotonically increase because
$q_{i,k}$ only appears in the denominator of $\Delta_{j,l}$.
That is, $\partial\Delta_{j,l}/\partial q_{i,k}\leq0$ for $(j,l)\neq(i,k)$.
Putting $\partial\Delta_{j,l}/\partial q_{i,k}\leq0$ for $(j,l)\neq(i,k)$
into equation \eqref{kkt1},  we immediately have $\partial\Delta_{i,k}/\partial q_{i,k}\cdot
\lambda_{i,k}\geq 1,\forall i,k$, which leads to $\lambda_{i,k}\neq 0$.
Furthermore, due to the complementary slackness $\lambda_{i,k}(2^{2\overline{R}_{i,k}}-\Delta_{i,k})=0$,
the equality $\Delta_{i,k}=2^{2\overline{R}_{i,k}}$ must hold for all stationary points.

Based on the above result, we can treat the first inequality constraint as equality, which would not change the stationary point of $\mathcal{P}1$ \cite[pp. 307]{30}
\begin{align}\label{P2-1}
&\frac{q_{i,k}|\mathbf{w}^{H}_{k}\mathbf{h}_{i,k}|^2}{\sum_{j=1}^2q_{j,k}|\mathbf{w}^{H}_{k}\mathbf{h}_{j,k}|^2}
+\frac{q_{i,k}|\mathbf{w}^{H}_k\mathbf{h}_{i,k}|^2}{\sum_{j=1}^2\sum_{l\neq k}q_{j,l}|\mathbf{w}^{H}_k\mathbf{h}_{j,l}|^2+\sigma^2_r}
\nonumber\\
&
=q_{i,k}|\mathbf{w}^{H}_{k}\mathbf{h}_{i,k}|^2\Big(
\frac{1}{\sum_{j=1}^2q_{j,k}|\mathbf{w}^{H}_{k}\mathbf{h}_{j,k}|^2}
\nonumber\\
&~~~~~~~~~~~~~~~~~~~~
+\frac{1}{\sum_{j=1}^2\sum_{l\neq k}q_{j,l}|\mathbf{w}^{H}_k\mathbf{h}_{j,l}|^2+\sigma^2_r}
\Big)
\nonumber\\
&=2^{2\overline{R}_{i,k}},~~\forall i,k.
\end{align}
Now, it is clear from \eqref{P2-1} that the term inside the parenthesis remains the same for both users pertaining to the $k^{\mathrm{th}}$ user pair.
Thus, dividing both sides of \eqref{P2-1} with $i=1$ by that with $i=2$ for the $k^{\mathrm{th}}$ user pair, the term inside the parenthesis can be cancelled and we get
$
q_{1,k}|\mathbf{w}^{H}_k\mathbf{h}_{1,k}|^2/(q_{2,k}|\mathbf{w}^{H}_k\mathbf{h}_{2,k}|^2)
=2^{2\overline{R}_{1,k}}/2^{2\overline{R}_{2,k}},
$
which implies that
\begin{align}
\frac{q_{i,k}|\mathbf{w}^{H}_k\mathbf{h}_{1,k}|^2}{\sum_{j=1}^2q_{j,k}|\mathbf{w}^H_{k}\mathbf{h}_{j,k}|^2}
=\frac{2^{2\overline{R}_{i,k}}}{2^{2\overline{R}_{1,k}}+2^{2\overline{R}_{2,k}}}
,~~\forall i,k.\nonumber
\end{align}
This completes the proof.

\section{Reformulation of $\mathcal{P}1$}

Based on \textbf{Property 1} and after some manipulations, the following problem $\mathcal{P}1'$ is equivalent to $\mathcal{P}1$ in the sense that they have the same optimal solutions
\begin{subequations}
\begin{align}
&\mathcal{P}1':\mathop{\mathrm{min}}_{\substack{\{\mathbf{v}_k,\mathbf{w}_k,p_k,\\q_{i,k},\beta_{i,k}\}}}
~\sum_{k=1}^K p_{k}+\sum_{i=1}^2\sum_{k=1}^K q_{i,k} \nonumber \\
&\mathrm{s.t.}~~q_{i,k}|\mathbf{w}^H_k\mathbf{h}_{i,k}|^2
\nonumber\\
&~~~~~
\geq\alpha_{i,k}\Big(\sum_{j=1}^2\sum_{l\neq k}q_{j,l}|\mathbf{w}^H_k\mathbf{h}_{j,l}|^2+\sigma_r^2\Big),~~\forall i,k
\label{13a}
\\
&~~~~~~
\beta_{i,k}\Big(\frac{p_{k}|\mathbf{g}^H_{i,k}\mathbf{v}_k|^2}{\theta_{i,k}}
-\sum_{l\neq k}p_{l}|\mathbf{g}^H_{i,k}\mathbf{v}_{l}|^2-\sigma_u^2
\Big)
\nonumber\\
&~~~~~
\geq  \sigma_z^2,~~\forall i,k
\label{13b}
\\
&~~~~~~
\eta(1-\beta_{i,k})\Big(\sum_{l=1}^Kp_{l}|\mathbf{g}^H_{i,k}\mathbf{v}_{l}|^2)+\sigma_u^2\Big)
\nonumber\\
&~~~~~
\geq q_{i,k}+2p_c-2E_{i,k},~~\forall i,k
\label{13c}
\\
&~~~~~~q_{i,k}\geq0,~\beta_{i,k}\in (0,1],~~\forall i,k
\nonumber\\
&~~~~~~
||\mathbf{v}_k||=1,~||\mathbf{w}_k||=1,~p_{k}\geq0,~~\forall k.
\end{align}
\end{subequations}
Since the fractional-quadratic function is known to be convex, the coupling between $\mathbf{w}_k$ and $q_{i,k}$ in \eqref{13a} can be resolved by letting $q_{i,k}=1/\xi_{i,k}$.
This is always possible because $\alpha_{i,k}$ defined in \eqref{alpha} satisfies $\alpha_{i,k}> 0$ and then from \eqref{13a} $q_{i,k}\neq 0$ holds.
On the other hand, to linearize the quadratic terms of $\mathbf{v}_k$ in \eqref{13b} and \eqref{13c},
we introduce new variables $\mathbf{V}_k:=p_{k}\mathbf{v}_k\mathbf{v}^H_k\succeq 0$ with $\mathrm{Rank}(\mathbf{V}_k)=1$ for all $k\in\mathcal{K}$.
Applying the transformations presented above, the problem $\mathcal{P}1'$ is equivalent to the following problem $\mathcal{P}1''$:
\begin{subequations}
\begin{align}
&\mathcal{P}1'':\mathop{\mathrm{min}}_{\substack{\{\mathbf{V}_k,\mathbf{w}_k,\\\xi_{i,k},\beta_{i,k}\}}}
~\sum_{k=1}^K \mathrm{Tr}(\mathbf{V}_{k})+\sum_{i=1}^2\sum_{k=1}^K \frac{1}{\xi_{i,k}} \nonumber \\
&\mathrm{s.t.}~~
\frac{|\mathbf{w}^H_k\mathbf{h}_{i,k}|^2}{\xi_{i,k}}
\nonumber\\
&~~~~~
\geq\alpha_{i,k}\Big(\sum_{j=1}^2\sum_{l\neq k}\frac{|\mathbf{w}^H_k\mathbf{h}_{j,l}|^2}{\xi_{j,l}}+\sigma_r^2\Big),~~\forall i,k
\label{P4a}
\\
&~~~~~~
\beta_{i,k}\Big(\frac{\mathrm{Tr}(\mathbf{\Theta}_{i,k}\mathbf{V}_k)}{\theta_{i,k}}-\mathop{\sum}_{l\neq k}\mathrm{Tr}(\mathbf{\Theta}_{i,k}\mathbf{V}_{l})-\sigma_u^2\Big)
\nonumber\\
&~~~~~
\geq\sigma^2_z
,~~\forall i,k
\label{P4b}
\\
&~~~~~~
\eta(1-\beta_{i,k})(\sum_{l=1}^K\mathrm{Tr}(\mathbf{\Theta}_{i,k}\mathbf{V}_{l})+\sigma_u^2)
\nonumber\\
&~~~~~\geq
\frac{1}{\xi_{i,k}}+2p_c-2E_{i,k}
,~~\forall i,k
\label{P4c}
\\
&~~~~~~\xi_{i,k}>0,~\beta_{i,k}\in (0,1],~~\forall i,k,
\nonumber\\
&~~~~~~
\mathbf{V}_k\succeq 0,~\mathrm{Rank}(\mathbf{V}_k)=1,~||\mathbf{w}_k||\leq1,~~\forall k,\label{P4d}
\end{align}
\end{subequations}
where $\mathbf{\Theta}_{i,k}=\mathbf{g}_{i,k}\mathbf{g}^{H}_{i,k}$.
Notice that the constraint $||\mathbf{w}_k||=1$ in $\mathcal{P}1'$ is relaxed into
$||\mathbf{w}_k||\leq1$, which does not affect our problem since the optimal $\mathbf{w}^*_k$ of $\mathcal{P}1''$ always satisfies $||\mathbf{w}^*_k||=1$.

To deal with \eqref{P4a}, we introduce slack variables $I_{j,l,k}\geq|\mathbf{w}^H_k\mathbf{h}_{j,l}|^2/\xi_{j,l}$, and \eqref{P4a} becomes the constraints \eqref{P4-a} and \eqref{P4-a2} in $\mathcal{P}2$.
Furthermore, the constraint \eqref{P4b} can be rearranged as an LMI through Schur Complement Lemma \cite{22}, and this gives \eqref{P4-b} in $\mathcal{P}2$.
Now we focus on \eqref{P4c}, which is non-convex in its current form.
However, by introducing a slack variable $\mu_{i,k}$ such that
\begin{align}\label{diff}
&\eta(1-\beta_{i,k})(\sum_{l=1}^K\mathrm{Tr}(\mathbf{\Theta}_{i,k}\mathbf{V}_{l})+\sigma_u^2)\geq\mu^2_{i,k}
\nonumber\\
&~~~~~~~~~~~~~~~~~~~~~~~~~~~~~~
\geq
\frac{1}{\xi_{i,k}}+2p_c-2E_{i,k},
\end{align}
the first and second inequalities of \eqref{diff} can be cast as \eqref{P4-c} and \eqref{P4-mu} in $\mathcal{P}2$, respectively.
Finally, it is remarkable that the constraints $0<\beta_{i,k}\leq1$ and $\xi_{i,k}\geq0$ in \eqref{P4d} are implicitly incorporated by the LMI constraints \eqref{P4-a2}-\eqref{P4-c} in $\mathcal{P}2$, and thus they can be dropped without changing the problem.
After the above procedure, $\mathcal{P}1''$ is equivalently transformed into $\mathcal{P}2$.

\section{Proof of Property 2}
We first address the proof for $\mathbf{w}^*$.
Assume that $\mathbf{w}=a_{1}\mathbf{h}_{1}+a_{2}\mathbf{h}_{2}+a_{3}\mathbf{h}_{\bot}$ with $\mathbf{h}^H_{\bot}\mathbf{h}_{i}=0$.
Consider two solutions of $\mathbf{w}$ in $\mathcal{P}3$, i.e., $\widehat{a}_1\mathbf{h}_{1}+\widehat{a}_2\mathbf{h}_{2}+\widehat{a}_3\mathbf{h}_{\bot}$ with $\widehat{a}_3\neq0$ and $\epsilon\widehat{a}_1\mathbf{h}_{1}+\epsilon\widehat{a}_2\mathbf{h}_{2}$, while other variables are fixed.
Based on the norm constraint $||\widehat{a}_{1}\mathbf{h}_{1}+\widehat{a}_{2}\mathbf{h}_{2}||^2+||\widehat{a}_{3}\mathbf{h}_{\bot}||^2
=\epsilon^2||\widehat{a}_{1}\mathbf{h}_{1}+\widehat{a}_{2}\mathbf{h}_{2}||^2$, we must have $\epsilon>1$.
On the other hand, putting the two solutions $\widehat{a}_1\mathbf{h}_1+\widehat{a}_2\mathbf{h}_{2}+
\widehat{a}_3\mathbf{h}_{\bot}$ and $\epsilon \widehat{a}_{1}\mathbf{h}_1+\epsilon \widehat{a}_{2}\mathbf{h}_2$ into the objective function of problem $\mathcal{P}3$, we have
\begin{align}
&\frac{\alpha_{i}\sigma_r^2}{|(\widehat{a}_{1}\mathbf{h}_{1}+\widehat{a}_{2}\mathbf{h}_{2}+\widehat{a}_{3}\mathbf{h}_{\bot})^H\mathbf{h}_{i}|^2}
=\frac{\alpha_{i}\sigma_r^2}{|(\widehat{a}_{1}\mathbf{h}_{1}+\widehat{a}_{2}\mathbf{h}_{2})^H\mathbf{h}_{i}|^2}
\nonumber\\
&
>
\frac{\alpha_{i}\sigma_r^2}{\epsilon^2|(\widehat{a}_{1}\mathbf{h}_{1}+\widehat{a}_{2}\mathbf{h}_{2})^H\mathbf{h}_{i}|^2}.\nonumber
\end{align}
This indicates that for any solution with $a_{3}\neq 0$, we can always find another one satisfying \eqref{26c} while giving smaller objective value.
Therefore, we must have $a^*_{3}=0$.

Next we address the proof for $\mathbf{V}^*$.
Since $\mathrm{Rank}(\mathbf{V})=1$, we can assume
$\mathbf{V}=[b_{1}\mathbf{g}_{1}+b_{2}\mathbf{g}_{2}+b_{3}\mathbf{g}_{\bot}][b_{1}\mathbf{g}_{1}+b_{2}\mathbf{g}_{2}+b_{3}\mathbf{g}_{\bot}]^H$
with $\mathbf{g}^H_{\bot}\mathbf{g}_{i}=0$.
Then, we have $\mathrm{Tr}(\mathbf{V})=||b_{1}\mathbf{g}_{1}+b_{2}\mathbf{g}_{2}||^2+||b_{3}\mathbf{g}_{\bot}||^2$.
With a similar proof to $\mathbf{w}^*$, it can be easily shown that $b_{3}=0$ would minimize the objective function.
Furthermore, due to $\mathbf{\Theta}_i\mathbf{g}_{\bot}=\mathbf{g}^H_{\bot}\mathbf{\Theta}_i=\mathbf{0}$, we also have
$\mathrm{Tr}(\mathbf{\Theta}_i\mathbf{V})=\mathrm{Tr}(\mathbf{\Theta}_i[b_{1}\mathbf{g}_{1}+b_{2}\mathbf{g}_{2}][b_{1}\mathbf{g}_{1}+b_{2}\mathbf{g}_{2}]^H)$,
meaning that $b_3$ does not affect the constraints \eqref{26b} and \eqref{26bb}.
This completes the proof.

\section{Proof of Property 3}
To prove part (i), consider the following inequality
\begin{align}\label{property4-1}
\Big(\frac{\mathbf{w}_k}{\xi_{i,k}}-\frac{\mathbf{w}^{[n]}_k}{\xi^{[n]}_{i,k}}\Big)^H\mathbf{h}_{i,k}\mathbf{h}^H_{i,k}
\Big(\frac{\mathbf{w}_k}{\xi_{i,k}}-\frac{\mathbf{w}^{[n]}_k}{\xi^{[n]}_{i,k}}\Big)\geq 0.
\end{align}
This always holds due to $\mathbf{h}_{i,k}\mathbf{h}^H_{i,k}\succeq 0$.
Then from \eqref{property4-1} we further have
\begin{align}
&-2\mathrm{Re}\Big(\frac{(\mathbf{w}^{[n]}_k)^H\mathbf{h}_{i,k}\mathbf{h}^H_{i,k}\mathbf{w}_k}{\xi^{[n]}_{i,k}\xi_{i,k}}\Big)+
\frac{(\mathbf{w}^{[n]}_k)^H\mathbf{h}_{i,k}\mathbf{h}^H_{i,k}\mathbf{w}^{[n]}_k}{(\xi^{[n]}_{i,k})^2}
\nonumber\\
&
\geq
\frac{\mathbf{w}^H_k\mathbf{h}_{i,k}\mathbf{h}^H_{i,k}\mathbf{w}_k}{\xi^2_{i,k}}.
\end{align}
Multiplying $\xi_{i,k}$ on both sides of the above inequality immediately yields
$\tilde{\Phi}^{[n]}_{i,k}(\mathbf{w}_k,\xi_{i,k})\geq \Phi_{i,k}(\mathbf{w}_k,\xi_{i,k})$.

To prove part (ii), we substitute $\mathbf{w}^{[n]}_k$ and $\xi^{[n]}_{i,k}$ into the definition of $\tilde{\Phi}^{[n]}_{i,k}$ in \eqref{gapp}, which yields
\begin{align}
\tilde{\Phi}^{[n]}_{i,k}(\mathbf{w}^{[n]}_k,\xi^{[n]}_{i,k})
=
-\frac{(\mathbf{w}^{[n]}_k)^H\mathbf{h}_{i,k}\mathbf{h}^H_{i,k}\mathbf{w}^{[n]}_k}{\xi^{[n]}_{i,k}}
=
\Phi_{i,k}(\mathbf{w}^{[n]}_k,\xi^{[n]}_{i,k})
.
\end{align}

To prove part (iii), we first calculate the following derivatives:
\begin{align}
&\partial\tilde{\Phi}^{[n]}_{i,k}/\partial\mathbf{w}_k=
-\Big[\frac{(\mathbf{w}^{[n]}_k)^H\mathbf{h}_{i,k}\mathbf{h}^H_{i,k}}{\xi^{[n]}_{i,k}}\Big]^T,
\nonumber\\
&
\partial \Phi_{i,k}/\partial\mathbf{w}_k=
-\Big[\frac{\mathbf{w}^H_k\mathbf{h}_{i,k}\mathbf{h}^H_{i,k}}{\xi^{[n]}_{i,k}}\Big]^T, \nonumber
\\
&\partial\tilde{\Phi}^{[n]}_{i,k}/\partial\xi_{i,k}=
\frac{(\mathbf{w}^{[n]}_k)^H\mathbf{h}_{i,k}\mathbf{h}^H_{i,k}\mathbf{w}^{[n]}_k}{(\xi^{[n]}_{i,k})^2},
\nonumber\\
&
\partial \Phi_{i,k}/\partial\xi_{i,k}=
\frac{\mathbf{w}^H_k\mathbf{h}_{i,k}\mathbf{h}^H_{i,k}\mathbf{w}_k}{(\xi_{i,k})^2}.
\end{align}
Then by putting $\mathbf{w}_k=\mathbf{w}^{[n]}_k$ and $\xi_{i,k}=\xi^{[n]}_{i,k}$ into the above equations, the proof for part (iii) is completed.

\section{CP-Free Initialization for $\{\mathbf{w}^{[0]}_k,\xi^{[0]}_{i,k}\}$}
\vspace{0.05in}
\noindent\underline{Optimizing $\mathbf{w}_{k}$ with $\xi_{i,k}$ fixed}
\vspace{0.05in}

When $\{\xi_{i,k}\}$ is fixed, problem \eqref{C1} is equivalent to
\begin{align}\label{eq gk}
\mathop{\mathrm{max}}_{||\mathbf{w}_k||=1}~\mathop{\mathrm{min}}_{i=1,2}\frac{|\mathbf{w}^H_k\mathbf{h}_{i,k}|^2/\xi_{i,k}}
{\alpha_{i,k}(\sum_{j=1}^2\sum_{l\neq k}|\mathbf{w}^H_k\mathbf{h}_{j,l}|^2/\xi_{j,l}+\sigma_r^2)}.
\end{align}
Now incorporating the constraint $\mathbf{w}^H_k\mathbf{w}_k=1$ to the noise variance $\sigma^2_r$ and defining $\mathbf{J}_k=\sum_{j}\sum_{l\neq k}\mathbf{h}_{j,l}\mathbf{h}^H_{j,l}
/\xi_{j,l}+\sigma_r^2\mathbf{I}$, problem \eqref{eq gk} becomes
\begin{align}
\mathop{\mathrm{max}}_{\mathbf{w}_k}~\mathop{\mathrm{min}}_{i=1,2}~
\frac{\mathbf{w}^H_k\mathbf{h}_{i,k}\mathbf{h}^{H}_{i,k}\mathbf{w}_k}{\alpha_{i,k}\xi_{i,k}\mathbf{w}^H_k\mathbf{J}_{k}\mathbf{w}_k}. \label{23}
\end{align}
With a further change of variable $\mathbf{u}_k:=\mathbf{J}^{1/2}_{k}\mathbf{w}_k$, problem \eqref{23} becomes
\begin{align}
\mathop{\mathrm{max}}_{\mathbf{u}_k}~\mathop{\mathrm{min}}_{i=1,2}~
\frac{
\mathbf{u}^H_k(\mathbf{J}^{-1/2}_{k}\mathbf{h}_{i,k}\mathbf{h}^{H}_{i,k}\mathbf{J}^{-1/2}_{k})\mathbf{u}_k}{
\alpha_{i,k}\xi_{i,k}\mathbf{u}^H_k\mathbf{u}_k},
\end{align}
which can be equivalently written as
\begin{align}\label{eq 2}
\mathop{\mathrm{max}}_{\mathbf{u}_k}~\mathrm{min}
\Big(|\mathbf{e}^H_{1,k}\mathbf{u}_k|,|\mathbf{e}^H_{2,k}\mathbf{u}_k|\Big)
~~\mathrm{s.t.}~||\mathbf{u}_k||=1,
\end{align}
where $\mathbf{e}_{i,k}=(1/\sqrt{\alpha_{i,k}\xi_{i,k}})\mathbf{J}^{-1/2}_{k}\mathbf{h}_{i,k}$.
From \eqref{eq 2}, it is clear that the optimal
$\mathbf{u}^*_k\in\mathrm{span}\{\mathbf{e}_{1,k},\mathbf{e}_{2,k}\}$.
Furthermore, based on \cite[Lemma 5-6]{24}, the optimal $\mathbf{u}^*_k$ can be expressed by $0\leq a_k\leq1$ as
\begin{align}\label{uk}
\mathbf{u}^*_k(a_k)=\sqrt{a_k}\cdot\frac{\mathbf{e}_{1,k}}{||\mathbf{e}_{1,k}||}+\sqrt{1-a_{k}}\cdot\mathrm{e}^{\mathrm{j}\angle(\mathbf{e}^H_{2,k}\mathbf{e}_{1,k})}
\frac{\mathbf{e}_{b,k}}{||\mathbf{e}_{b,k}||},
\end{align}
where
$
\mathbf{e}_{b,k}=\mathbf{e}_{2,k}-
(\mathbf{e}^H_{1,k}\mathbf{e}_{2,k}/||\mathbf{e}_{1,k}||^2)\cdot\mathbf{e}_{1,k}
$.
Putting \eqref{uk} back to \eqref{eq 2},
\eqref{eq 2} is reduced into
\begin{align}\label{eq 3}
&\mathop{\mathrm{max}}_{a_k\in[0,1]}~\mathrm{\mathop{min}}\Big(
\sqrt{a_k}||\mathbf{e}_{1,k}||,\sqrt{a_k}\frac{|\mathbf{e}^H_{2,k}\mathbf{e}_{1,k}|}{||\mathbf{e}_{1,k}||}+\sqrt{1-a_k}||\mathbf{e}_{b,k}||
\Big).
\end{align}
Inside the min-function, it is clear that the term $\sqrt{a_k}||\mathbf{e}_{1,k}||$ is an increasing function of $a_k$,
and thus the maximum of $\sqrt{a_k}||\mathbf{e}_{1,k}||$ is obtained when $a_k=1$.
On the other hand, taking the derivative with respect to $a_k$,
it can be shown that $\sqrt{a_k}\frac{|\mathbf{e}^H_{2,k}\mathbf{e}_{1,k}|}{||\mathbf{e}_{1,k}||}+\sqrt{1-a_k}||\mathbf{e}_{b,k}||$
is an increasing function of $a_k$ when $a_k\leq\frac{|\mathbf{e}^H_{2,k}\mathbf{e}_{1,k}|^2}{||\mathbf{e}_{1,k}||^2||\mathbf{e}_{2,k}||^2}$,
and a decreasing function of $a_k$ when $a_k\geq\frac{|\mathbf{e}^H_{2,k}\mathbf{e}_{1,k}|^2}{||\mathbf{e}_{1,k}||^2||\mathbf{e}_{2,k}||^2}$.
Therefore, the maximum of $\sqrt{a_k}\frac{|\mathbf{e}^H_{2,k}\mathbf{e}_{1,k}|}{||\mathbf{e}_{1,k}||}+\sqrt{1-a_k}||\mathbf{e}_{b,k}||$ is obtained when $a_k=\frac{|\mathbf{e}^H_{2,k}\mathbf{e}_{1,k}|^2}{||\mathbf{e}_{1,k}||^2||\mathbf{e}_{2,k}||^2}$.

Based on the above analysis, it can be shown that the optimal $a^*_k$ to problem \eqref{eq 3}
must be $a^*_k\in\{1,\frac{|\mathbf{e}^H_{2,k}\mathbf{e}_{1,k}|^2}{||\mathbf{e}_{1,k}||^2||\mathbf{e}_{2,k}||^2},a^{\mathrm{int}}_{k}\}$, where $a^{\mathrm{int}}_{k}$ is the intersection point of
$\sqrt{a_k}||\mathbf{e}_{1,k}||$ and $\sqrt{a_k}\frac{|\mathbf{e}^H_{2,k}\mathbf{e}_{1,k}|}{||\mathbf{e}_{1,k}||}+\sqrt{1-a_k}||\mathbf{e}_{b,k}||$.
Therefore, the optimal $a^*_k$ can be chosen from the three points by comparing their objective values in \eqref{eq 3}.
Notice that to compute the intersection point $a^{\mathrm{int}}_{k}$, we have
\begin{align}\label{37}
&
||\mathbf{e}_{1,k}||\sqrt{a_k}=\frac{|\mathbf{e}^H_{2,k}\mathbf{e}_{1,k}|}{||\mathbf{e}_{1,k}||}\sqrt{a_k}+||\mathbf{e}_{b,k}||\sqrt{1-a_k},
\end{align}
which leads to
\begin{align}
  &a^{\mathrm{int}}_{k}=
  \frac{||\mathbf{e}_{b,k}||^2}{
  (||\mathbf{e}_{1,k}||-\frac{|\mathbf{e}^H_{2,k}\mathbf{e}_{1,k}|}{||\mathbf{e}_{1,k}||})^2+||\mathbf{e}_{b,k}||^2
  }
\nonumber.
 \end{align}
when $||\mathbf{e}_{1,k}||\geq|\mathbf{e}^H_{2,k}\mathbf{e}_{1,k}|/||\mathbf{e}_{1,k}||$;
otherwise the intersection point does not exist.

Putting the value of $a^*_k$ into \eqref{uk}, we obtain $\mathbf{u}^*_{k}$, and the optimal $\mathbf{w}^*_k$ to problem \eqref{eq gk} can be recovered as
$
\mathbf{w}^*_k=\mathbf{J}^{-1/2}_k\mathbf{u}^*_{k}\Big/||\mathbf{J}^{-1/2}_k\mathbf{u}^*_{k}||$.

\vspace{0.05in}
\noindent\underline{Optimizing $\xi_{i,k}$ with $\mathbf{w}_{k}$ fixed}
\vspace{0.05in}

When $\mathbf{w}_{k}$ is fixed, define the signal-term nonnegative matrix $\mathbf{D}\in \mathbb{R}^{2K\times 2K}_{+}$ as
\begin{align}
&\mathbf{D}=\mathrm{diag}\Big[
\frac{\alpha_{1,1}}{|\mathbf{w}^H_{1}\mathbf{h}_{1,1}|^2},
\frac{\alpha_{2,1}}{|\mathbf{w}^H_{1}\mathbf{h}_{2,1}|^2},
...,
\nonumber\\
&~~~~~~~~~~~~~~
\frac{\alpha_{1,K}}{|\mathbf{w}^H_{K}\mathbf{h}_{1,K}|^2},
\frac{\alpha_{2,K}}{|\mathbf{w}^H_{K}\mathbf{h}_{2,K}|^2}\Big],
\end{align}
and the cross-talk nonnegative matrix $\mathbf{R}\in \mathbb{R}^{2K\times 2K}_{+}$ with
the ${(2l-2+j,2k-2+i)}^{\mathrm{th}}$ element for $i,j=1,2$ and $l,k=1,...,K$
being
$|\mathbf{w}^H_{k}\mathbf{h}_{j,l}|^2$  when $l\neq k$, and $0$ otherwise.
Following the derivation of (17) in \cite{27}, the objective function of \eqref{C1} is maximized when $[1/\xi^*_{1,1},1/\xi^*_{1,2},...,1/\xi^*_{1,K},1/\xi^*_{2,K},1]^T$ is the dominant eigenvector of
\begin{align}
&
\left[
\begin{array}{cccc}
\mathbf{D}\mathbf{R}^H & \sigma_r\mathbf{D}\mathbf{1}_{2K}^H
\\
\mathbf{1}_{2K}^H\mathbf{D}\mathbf{R}^H/P & \sigma_r\mathbf{1}_{2K}^H\mathbf{D}\mathbf{1}_{2K}/P
\end{array}
\right].
\end{align}

\vspace{0.05in}
\noindent\underline{Generating $\{\mathbf{w}^{[0]}_k,\xi^{[0]}_{i,k}\}$}
\vspace{0.05in}

Starting from $\{\xi_{i,k}\}$ with $\sum_{i}\sum_{k}1/\xi_{i,k}=P$, $\mathbf{w}_k$ and $\xi_{i,k}$ are updated iteratively until a feasible $\{\mathbf{w}'_k,\xi'_{i,k}\}$ is found.
Then we
set $\{\mathbf{w}^{[0]}_k=\mathbf{w}'_k\}$, and by solving
\begin{align}
\frac{|(\mathbf{w}^{[0]}_k)^H\mathbf{h}_{i,k}|^2}{\xi_{i,k}}=\alpha_{i,k}\Big(\sum_{j=1}^2\sum_{l\neq k}\frac{|(\mathbf{w}^{[0]}_k)^H\mathbf{h}_{j,l}|^2}{\xi_{j,l}}+\sigma_r^2\Big),
\end{align}
we obtain $[1/\xi^{[0]}_{1,1},1/\xi^{[0]}_{1,2},...,1/\xi^{[0]}_{1,K},1/\xi^{[0]}_{2,K}]^T=
\sigma_r^2\Big(\mathbf{I}_{2K}-\mathbf{D}\mathbf{R}^H\Big)^{-1}\mathbf{D}\mathbf{1}_{2K}
$.

\section{Derivation of $p^*_{i,k}$ for $\mathcal{P}6$}
After putting $\mathbf{w}^*_{i,k},\mathbf{v}^*_{i,k}$ and $q^*_{i,k}$ into $\mathcal{P}6$, $\mathcal{P}6$ reduces into
\begin{align}
&\mathop{\mathrm{min}}_{\substack{\{p_{i,k}\geq0,\beta_{i,k}\in(0,1]\}}}
~\sum_{i=1}^2\sum_{k=1}^K p_{i,k} \nonumber\\
&\mathrm{s.t.}~~
\frac{\beta_{i,k}p_{i,k}N\varrho_{i,k}}
{\beta_{i,k}\sigma_u^2+\sigma_z^2}
\geq \theta_{i,k},~~
\forall i,k
\nonumber \\
&\eta(1-\beta_{i,k})\Big(p_{i,k}N\varrho_{i,k}+\sigma_u^2\Big)\geq \frac{\theta_{3-i,k}\sigma^2_r}{N\varrho_{i,k}}
+2p_c-2E_{i,k},
~~
\forall i,k,
\nonumber
\end{align}
where we have used $||\mathbf{g}_{i,k}||^2=N\varrho_{i,k}$.
Now we consider two cases for solving the above problem.
When $\frac{\theta_{3-i,k}\sigma^2_r}{N\varrho_{i,k}}
+2p_c-2E_{i,k}\leq0$, the second constraint is always satisfied, and the first constraint can be rewritten as
$
p_{i,k}\geq\frac{\theta_{i,k}}{N\varrho_{i,k}}(\sigma_u^2+\frac{\sigma_z^2}{\beta_{i,k}})
$.
Therefore, to minimize $p_{i,k}$, we need to maximize $\beta_{i,k}$.
Since $\beta_{i,k}\in(0,1]$, we have $\beta^*_{i,k}=1$.
Putting $\beta^*_{i,k}=1$ into $
p_{i,k}\geq\frac{\theta_{i,k}}{N\varrho_{i,k}}(\sigma_u^2+\frac{\sigma_z^2}{\beta_{i,k}})
$ and since $p_{i,k}$ is a variable to minimize,
we obtain the first line of \eqref{pik}.

On the other hand when $\frac{\theta_{3-i,k}\sigma^2_r}{N\varrho_{i,k}}+2p_c-2E_{i,k}>0$,
the two constraints can be rewritten as
$
p_{i,k}\geq\frac{\theta_{i,k}}{N\varrho_{i,k}}(\sigma_u^2+\frac{\sigma_z^2}{\beta_{i,k}})
$
and
$
p_{i,k}
\geq
\frac{\theta_{3-i,k}\sigma^2_r(N\varrho_{i,k})^{-1}+2p_c-2E_{i,k}}{\eta(1-\beta_{i,k})N\varrho_{i,k}}
-\frac{\sigma_u^2}{N\varrho_{i,k}}
$, respectively.
Combining the above two inequalities,
we have
\begin{align}\label{A2}
&
p_{i,k}
\geq\mathop{\mathrm{max}}\Big(\frac{\theta_{i,k}}{N\varrho_{i,k}}(\sigma_u^2+\frac{\sigma_z^2}{\beta_{i,k}}),
\nonumber\\
&~~~~~~~~~~
\frac{\theta_{3-i,k}\sigma^2_r(N\varrho_{i,k})^{-1}+2p_c-2E_{i,k}}{\eta(1-\beta_{i,k})N\varrho_{i,k}}
-\frac{\sigma_u^2}{N\varrho_{i,k}}
\Big).
\end{align}
Inside the max function of \eqref{A2}, the first term is a decreasing function of $\beta_{i,k}$ while the second term is an increasing function of $\beta_{i,k}$.
Therefore, the minimum of $p_{i,k}$ is obtained when
\begin{align}
&
\theta_{i,k}(\sigma_u^2+\frac{\sigma_z^2}{\beta_{i,k}})=
\frac{\theta_{3-i,k}\sigma^2_r(N\varrho_{i,k})^{-1}+2p_c-2E_{i,k}}{\eta(1-\beta_{i,k})}
-\sigma_u^2
\nonumber
\end{align}
which leads to
\begin{align}
&\beta^*_{i,k}=
\frac{
2\theta_{i,k}\sigma^2_z}
{B_{i,k}
+\sqrt{B^2_{i,k}+4\theta_{i,k}(\theta_{i,k}+1)\sigma^2_u\sigma^2_z}},
\nonumber
\end{align}
with $B_{i,k}=\theta_{i,k}\sigma^2_z-(\theta_{i,k}+1)\sigma^2_u+[\theta_{3-i,k}\sigma^2_r(N\varrho_{i,k})^{-1}+2p_c-2E_{i,k}]/\eta$.
Putting $\beta^*_{i,k}$ into \eqref{A2} and since $p_{i,k}$ is a variable to minimize, we obtain the second line of \eqref{pik}.


\begin{thebibliography}{31}

\bibitem{1}L.-R. Varshney, ``Transporting information and energy simultaneously,'' in \emph{Proc. IEEE ISIT}, Jul. 2008, pp. 1612-1616.

\bibitem{J2}D.-S. Michalopoulos, H.-A. Suraweera, and R. Schober, ``Relay selection for simultaneous information transmission and wireless energy transfer: a tradeoff perspective,'' \emph{IEEE J. Sel. Area Commun.}, vol. 33, no. 8, pp. 1578-1594, Aug. 2015.

\bibitem{2}M. Xia and S. A\"{i}ssa, ``On the efficiency of far-field wireless power transfer,'' \emph{IEEE Trans. Signal Process.}, vol. 63, no. 11, pp. 2835-2847, Jun. 2015.

\bibitem{J4}A.-G. Rodriguez, and C. Masouros, ``Power-efficient Tomlinson-Harashima precoding for the downlink of multi-user MISO systems,''
\emph{IEEE Trans. Commun.}, vol. 62, no. 6, pp. 1884-1896, Jun. 2014.

\bibitem{3}R. Zhang, and C.-K. Ho, ``MIMO broadcasting for simultaneous wireless information and power transfer,'' \emph{IEEE Trans. Wireless Commun.}, vol. 12, no. 5, pp. 1989-2001, May 2013.

\bibitem{5}S. Timotheou, I. Krikidis, G. Zheng, and B. Ottersten, ``Beamforming for MISO interference channels with QoS and RF energy transfer,'' \emph{IEEE Trans. Wireless Commun.}, vol. 13, no. 5, pp. 2646-2658, May 2014.

\bibitem{J1}S. Timotheou, I. Krikidis, S. Karachontzitis, and K. Berberidis, ``Spatial domain simultaneous information and power transfer for MIMO channels,'' \emph{IEEE Trans. Wireless Commun.}, vol. 14, no. 8, pp. 4115-4128, Aug. 2015.

\bibitem{J3}D.-W.-K. Ng and R. Schober, ``Secure and green SWIPT in distributed antenna networks with limited backhaul capacity,''
\emph{IEEE Trans. Wireless Commun.}, vol. 14, no. 9, pp. 5082-5097, Sep. 2015.

\bibitem{8}H. Ju and R. Zhang, ``Throughput maximization in wireless powered communication networks,'' \emph{IEEE Trans. Wireless Commun.}, vol. 13, no. 1, pp. 418-428, Jan. 2014.

\bibitem{11}L. Liu, R. Zhang, and K.-C. Chua, ``Multi-antenna wireless powered communication with energy beamforming,'' \emph{IEEE Trans. Commun.}, vol. 62, no. 12, pp. 4349-4361, Dec. 2014.

\bibitem{12}Z. Wen, S. Wang, C. Fan, and W. Xiang, ``Joint transceiver and power splitter design over two-way relaying channel with lattice codes and energy harvesting,'' \emph{IEEE Commun. Lett.}, vol. 18, no. 11, pp. 2039-2042, Nov. 2014.

 \bibitem{14}M. Tao, and R. Wang, ``Linear precoding for multi-pair two-way MIMO relay systems with max-min fairness,'' \emph{IEEE Trans. Signal Process.}, vol. 60, no. 10, pp. 5361-5370, Oct. 2012.

\bibitem{13}A. Sezgin, H. Boche, and A.-S. Avestimehr, ``Bidirectional multi-pair network with a MIMO relay: beamforming strategies and lack of duality,'' in \emph{Proc. Annu. Allerton Conf. Commun., Control, Comput.}, Monticello, USA, pp. 72-77, 2010.

\bibitem{17}W. Nam, S. Chung, and Y.-H. Lee, ``Capacity of the Gaussian two-way relay channel to within 1/2 bit,'' \emph{IEEE Trans. Inf. Theory}, vol. 56, no. 11, pp. 5488-5494, Nov. 2010.

\bibitem{15}X. Yuan, T. Yang, and I.-B. Collings, ``Multiple-input multiple-output two-way relaying: a space-division approach,'' \emph{IEEE Trans. Inf. Theory}, vol. 59, no. 10, pp. 6421-6440, Oct. 2013.

\bibitem{16}E. Karipidis, N.-D. Sidiropoulos, and Z.-Q. Luo, ``Quality of service and max-min fair transmit beamforming to multiple cochannel multicast groups,'' \emph{IEEE Trans. Signal Process.}, vol. 56, no. 3, pp. 1268-1279, Mar. 2008.

\bibitem{20}A. Ben-Tal and A. Nemirovski, \emph{Lectures on Modern Convex Optimization}. MPS/SIAM Series on Optimizations, 2013.

\bibitem{30}D.-P. Bertsekas, \emph{Nonlinear Programming}. Athena Scientific Optimization and Computation Series, 1999.

\bibitem{24}K. Xiong, P. Fan, Z. Xu, H. Yang, and K.-B. Letaief, ``Optimal cooperative beamforming design for MIMO decode-and-forward relay channels,'' \emph{IEEE Trans. Signal Process.}, vol. 62, no. 6, pp. 1476-1489, Mar. 2014.

\bibitem{21}Z.-Q. Luo, W.-K. Ma, A.-M.-C. So, Y. Ye, and S. Zhang, ``Semidefinite relaxation of quadratic optimization problems,'' \emph{IEEE Signal Process. Mag.}, vol. 27, no. 3, pp. 20-34, May 2010.

\bibitem{23}Y. Huang, and D.-P. Palomar, ``Rank-constrained separable semidefinite programming with applications to optimal beamforming,'' \emph{IEEE Trans. Signal Process.}, vol. 58, no. 2, pp. 664-678, Feb. 2010.

\bibitem{22}S. Boyd and L. Vandenberghe, \emph{Convex Optimization}. Cambridge, U.K.: Cambridge Univ. Press, 2004.

\bibitem{28}X. Wu, W.-K. Ma, and A.-M.-C. So, ``Physical-layer multicasting by stochastic transmit beamforming and Alamouti space-time coding,'' \emph{IEEE Trans. Signal Process.}, vol. 61, no. 17, pp. 4230-4245, Sep. 2013.

\bibitem{B1}T. Lipp and S. Boyd, ``Variations and extensions of the convex-concave procedure,'' \emph{Optimization and Engineering}, vol. 17, no. 2, pp. 263-287, Jun. 2016.

\bibitem{A6}Y. Cheng, and M. Pesavento, ``Joint optimization of source power allocation and distributed relay beamforming in multiuser peer-to-peer relay networks,'' \emph{IEEE Trans. Signal Process.}, vol. 60, no. 6, pp. 2962-2973, Jun. 2012.

\bibitem{A1}A. Schad, K.-L. Law, and M. Pesavento, ``Rank-two beamforming and power allocation in multicasting relay networks,'' \emph{IEEE Trans. Signal Process.}, vol. 63, no. 13, pp. 3435-3447, Jul. 2015.

\bibitem{25}B.-R. Marks, and G.-P. Wright, ``A general inner approximation algorithm for nonconvex mathematical programs,'' \emph{Operation Research}, vol. 26, no. 4, Jul. 1978.

\bibitem{B2}G. Zheng, ``Joint beamforming optimization and power control for full-duplex MIMO two-way relay channel,''
\emph{IEEE Trans. Signal Process.}, vol. 63, no. 3, pp. 555-566, Feb. 2015.


\bibitem{27}M. Schubert and H. Boche, ``Solution of the multiuser downlink beamforming problem with individual SINR constraints,'' \emph{IEEE Trans. Veh. Technol.}, vol. 53, no. 1, pp. 18-28, Jan. 2004.

\bibitem{19}T.-K. Moon and W.-C. Stirling, \emph{Mathematical Methods and Algorithms for Signal Processing}. New Jersey, U.S.: Prentice Hall. 2000.

\bibitem{C1}K. Huang and E.-G. Larsson, ``Simultaneous information and power transfer for broadband wireless systems,'' \emph{IEEE Trans. Signal Process.}, vol. 61, no. 23, pp. 5972-5986, Dec. 2013.

\bibitem{R4}S. Wang, Y.-C. Wu, and M. Xia, ``Achieving global optimality for wirelessly-powered multi-antenna TWRC with lattice codes,'' in \emph{Proc. IEEE ICASSP}, Shanghai, China, Mar. 2016, pp. 3556-3560.

\bibitem{R1}X. Lu, P. Wang, D. Niyato, D.-I. Kim, and Z. Han, ``Wireless networks with RF energy harvesting: a contemporary survey,'' \emph{IEEE Commun. Surveys Tuts}, vol. 17, no. 2, pp. 757-789, Second Quart. 2015.

\bibitem{R2}H.-Q. Ngo, E.-G. Larsson, and T.-L. Marzetta, ``Energy and spectral efficiency of very large multiuser MIMO systems,'' \emph{IEEE Trans. Commun.}, vol. 61, no. 4, pp. 1436-1449, Apr. 2013.

\bibitem{R3}
S. Kashyap, E. B\"{j}ornson, and E.-G. Larsson, ``Can wireless power transfer benefit from large transmitter arrays?'' in \emph{Proc. IEEE WPTC}, Boulder, CO, May 2015, pp. 1-3.

\bibitem{R5}
A. Beck, ``On the convergence of alternating minimization for convex programming with applications to iteratively reweighted least squares and decomposition schemes,'' \emph{SIAM J. Optim.}, vol. 25, no. 1, pp. 185-209, Jan. 2015.
\end{thebibliography}
\end{document}